\documentclass[10pt,onecolumn]{IEEEtran}
\usepackage{soul}
\usepackage[mathscr]{eucal}
\usepackage[cmex10]{amsmath}
\usepackage{epsfig,epsf}
\usepackage{amssymb,amsmath,amsthm,amsfonts,latexsym}
\usepackage{amsmath,graphicx,bm,xcolor,url,overpic}
\usepackage{fixltx2e}
\usepackage{array}
\usepackage{verbatim}
\usepackage{bm}
\usepackage{algorithmic}
\usepackage{algorithm}
\usepackage{verbatim}
\usepackage{textcomp}
\usepackage{mathrsfs}
\usepackage{epstopdf}

\newcommand{\openone}{\leavevmode\hbox{\small1\normalsize\kern-.33em1}}

\catcode`~=11 \def\UrlSpecials{\do\~{\kern -.15em\lower .7ex\hbox{~}\kern .04em}} \catcode`~=13

\allowdisplaybreaks[4]

\newcommand{\nn}{\nonumber}

\newcommand{\calB}{\mathcal{B}}

\newcommand{\calK}{\mathcal{K}}
\newcommand{\calL}{\mathcal{L}}
\newcommand{\calM}{\mathcal{M}}

\newcommand{\calP}{\mathcal{P}}

\newcommand{\calR}{\mathcal{R}}

\newcommand{\calU}{\mathcal{U}}

\newcommand{\calX}{\mathcal{X}}

\newcommand{\bT}{\mathbf{T}}

\newcommand{\bz}{\mathbf{z}}
\newcommand{\bZ}{\mathbf{Z}}

\newcommand{\rmb}{\mathrm{b}}

\newcommand{\rmc}{\mathrm{c}}

\newcommand{\bbE}{\mathbb{E}}

\newcommand{\bbN}{\mathbb{N}}

\newcommand{\bbR}{\mathbb{R}}

\DeclareMathAlphabet{\mathbsf}{OT1}{cmss}{bx}{n}
\DeclareMathAlphabet{\mathssf}{OT1}{cmss}{m}{sl}

\DeclareSymbolFont{bsfletters}{OT1}{cmss}{bx}{n}
\DeclareSymbolFont{ssfletters}{OT1}{cmss}{m}{n}
\DeclareMathSymbol{\bsfGamma}{0}{bsfletters}{'000}
\DeclareMathSymbol{\ssfGamma}{0}{ssfletters}{'000}
\DeclareMathSymbol{\bsfDelta}{0}{bsfletters}{'001}
\DeclareMathSymbol{\ssfDelta}{0}{ssfletters}{'001}
\DeclareMathSymbol{\bsfTheta}{0}{bsfletters}{'002}
\DeclareMathSymbol{\ssfTheta}{0}{ssfletters}{'002}
\DeclareMathSymbol{\bsfLambda}{0}{bsfletters}{'003}
\DeclareMathSymbol{\ssfLambda}{0}{ssfletters}{'003}
\DeclareMathSymbol{\bsfXi}{0}{bsfletters}{'004}
\DeclareMathSymbol{\ssfXi}{0}{ssfletters}{'004}
\DeclareMathSymbol{\bsfPi}{0}{bsfletters}{'005}
\DeclareMathSymbol{\ssfPi}{0}{ssfletters}{'005}
\DeclareMathSymbol{\bsfSigma}{0}{bsfletters}{'006}
\DeclareMathSymbol{\ssfSigma}{0}{ssfletters}{'006}
\DeclareMathSymbol{\bsfUpsilon}{0}{bsfletters}{'007}
\DeclareMathSymbol{\ssfUpsilon}{0}{ssfletters}{'007}
\DeclareMathSymbol{\bsfPhi}{0}{bsfletters}{'010}
\DeclareMathSymbol{\ssfPhi}{0}{ssfletters}{'010}
\DeclareMathSymbol{\bsfPsi}{0}{bsfletters}{'011}
\DeclareMathSymbol{\ssfPsi}{0}{ssfletters}{'011}
\DeclareMathSymbol{\bsfOmega}{0}{bsfletters}{'012}
\DeclareMathSymbol{\ssfOmega}{0}{ssfletters}{'012}

\newcommand{\tilk}{\tilde{k}}

\newcommand{\tilP}{\tilde{P}}

\newcommand{\tils}{\tilde{s}}

\newcommand{\hatx}{\hat{x}}
\newcommand{\hatX}{\hat{X}}
\newcommand{\tilx}{\tilde{x}}
\newcommand{\tilX}{\tilde{X}}

\newtheorem{theorem}{Theorem}
\newtheorem{lemma}[theorem]{Lemma}

\newtheorem{corollary}[theorem]{Corollary}
\newtheorem{definition}{Definition}

\theoremstyle{remark}

\graphicspath{{./figures/}}
\usepackage{graphicx,cite}
\usepackage{epstopdf}
\usepackage{tikz}
\usepackage{stfloats}
\usepackage{url}
\usetikzlibrary{arrows.meta, positioning, shapes, calc, fit,backgrounds,shadows}
\usepackage{enumerate}
\usepackage{bbm}
\usepackage{graphicx}
\usepackage{caption}
\usepackage[ colorlinks = true,
linkcolor = blue,
urlcolor  = blue,
citecolor = red,
anchorcolor = green,]{hyperref}
\usepackage{soul,color}
\usepackage{multirow}
\usepackage{amsthm}
\theoremstyle{plain}

\IEEEoverridecommandlockouts
\soulregister\em7
\soulregister\cite7
\soulregister\ref7
\soulregister\eqref7
\soulregister\underline7
\soulregister\emph7
\soulregister\footnote7

\definecolor{Dyellow}{RGB}{254,152,0}
\definecolor{Dgreen}{RGB}{0,176,80}

\title{Successive Refinement Under Strong-Sense Perfect Perception}
\author{Yu Yang, Changhong Liu, Weijie Yuan, and Lin Zhou%
\thanks{The authors are with the School of Automation and Intelligent Manufacturing,
Southern University of Science and Technology, Shenzhen, China.}}

\begin{document}
\maketitle

\begin{abstract}
We revisit a multiterminal lossy source coding problem named successive refinement and derive the rate-distortion-perception region under the strong-sense perfect perception constraint in the presence of unlimited common randomness. Specifically, in successive refinement, one aims to compress a source sequence and allows two distinct decoders to recover the source sequence at different distortion levels. By imposing the strong-sense perfect perception constraint, our results refine the previous result by analyzing the impact of the perceptual quality. Our achievability proof is inspired by output constrained lossy source coding and our converse proof adapts the proof steps of the standard successive refinement problem. Furthermore, we provide a numerical example of the Bernoulli source to illustrate our result and show that the Bernoulli source under Hamming distortion is successively refinable even with the strong-sense perfect perception constraint.
\end{abstract}


\section{Introduction}
\label{sec:intro}
Lossy compression, also known as rate-distortion (RD), reduces the number of bits needed to compress a source sequence by allowing the reconstructed sequence to be different from the source sequence, where the difference is characterized via a distortion level. Classical RD theory characterizes the asymptotic minimum achievable rate under a prescribed distortion level~\cite{Shannon_1959_Coding-theorems-for-a-discrete-source-with-a-fidelity-criterion}. However, for certain applications including image compression, reducing the distortion level does not always improve the reconstruction quality. As observed by Blau and Michaeli, using the RD theory could lead to image reconstructions that are either blurred or lack natural details~\cite{Blau_2018_The-Perception-Distortion-Tradeoff,Blau_2019_The-Rate-Distortion-Perception-Tradeoff}, which leads to poor perceptual quality for humans. The above limitation has become particularly severe in generative compression, where the perceptual quality is treated as an explicit design objective beyond the compression rate and the distortion level, especially at low bit rates \cite{Agustsson_2019_Generative_Adversarial_Networks_for_Extreme_Learned_Image_Compression,Mentzer_2020_High-Fidelity-Generative-Image-Compression}.

To refine the RD theory, Blau and Michaeli pioneered the rate-distortion-perception (RDP) theory~\cite{Blau_2018_The-Perception-Distortion-Tradeoff,Blau_2019_The-Rate-Distortion-Perception-Tradeoff} and proposed the RDP function as the corresponding minimal compression rate. In particular, the RDP theory introduces an additional perception constraint that compares the distributions of the source sequence and the reconstructed version. The distortion constraint evaluates how closely a reconstructed sequence represents its corresponding source sequence, while the perception constraint evaluates whether the reconstructed sequence has similar statistical behavior as the source sequence. A lossy image compression system should balance the compression rate, the distortion level, and the perceptual quality. When the reconstructed sequence has the same distribution as the source distribution, the compression is said to achieve perfect perception.

Although the RDP function was proposed in~\cite{Blau_2018_The-Perception-Distortion-Tradeoff,Blau_2019_The-Rate-Distortion-Perception-Tradeoff}, its operational meaning was subsequently revealed by Theis and Wagner~\cite{Theis_2021_A-coding-theorem-for-the-rate-distortion-perception-function}, who used randomized coding schemes to show that the RDP function is achievable for lossy compression with both the distortion and perception constraints. Subsequently, Chen \emph{et al.} studied the RDP theory under different perception constraints and randomness assumptions \cite{chen_2022_on-the-rate-distortion-perception-function}. In particular, the authors of \cite{chen_2022_on-the-rate-distortion-perception-function} distinguished two notions of perfect perception. Weak-sense perfect perception requires each reconstructed  source symbol to have the same distribution as each corresponding source symbol, whereas strong-sense perfect perception requires the entire reconstructed source sequence to have the same distribution as the original source sequence. Furthermore, Chen \emph{et al.} showed that, when unlimited common randomness is shared between the encoder and decoder, the strong-sense perception incurs no rate penalty compared with the weak-sense perception. We should like to comment that the RDP theory is closely related to output constrained lossy source coding studied by Saldi \emph{et al.}~\cite{Saldi_2015_Output_Constrained,Saldi_2015_Randomized_quantization}.

Despite being insightful, the above results were restricted to the point-to-point (P2P) setting with a single encoder and decoder. However, in many practical applications, one needs to serve multiple decoders with different reconstruction requirements. For example, recent progressive learned image compression methods allow a decoder to obtain an initial reconstruction from part of the compressed bits and improve its quality as additional compressed bits are received \cite{Hojjat_2023_ProgDTD-Progressive-Learned-Image-Compression-with-Double-Tail-Drop-Training,Presta_2025_Efficient-Progressive-Image-Compression-with-Variance-Aware-Masking}. The corresponding information theoretic model for this scenario is successive refinement (SR)~\cite{EquitzC_1991_Successive-refinement-of-information,Rimoldi_1994_Successive-refinement-of-information-characterization-of-the-achievable-rates}. In this model, a source sequence is compressed into a base-layer message and a refinement-layer message. The base-layer message allows the first decoder to produce a coarse reconstruction, while the second decoder uses both messages to produce a reconstruction with a lower distortion. Such a layered coding framework enables the same encoder to support multiple decoders with different distortion levels.

Recently, the RDP theory has been extended to the SR under the weak-sense perfect perception constraint by Zhang \emph{et al.}~\cite{Zhang_2025_Universal-Rate-Distortion-Perception-Representations-for-Lossy-Compression}. However, the strong-sense perfect perception case has not yet been addressed. In this paper, we fill the above research gap and characterize the first-order asymptotic RDP region for SR under strong-sense perfect perception with unlimited common randomness. Our main result shows that the classical two-layer rate structure is preserved, with the perception requirements appearing only as restrictions on the admissible reconstruction distributions. Our technical contribution lies in the achievability proof. Specifically, we develop a superposition-based output synthesis scheme for the two-layer compression architecture of SR. A layered soft covering argument approximates the source distribution at the classical SR rates, while a maximal coupling correction enforces the exact distribution required by strong-sense perfect perception with vanishing distortion loss and no additional compression rate. Furthermore, for any Bernoulli source under Hamming distortion, we calculate the explicit RDP region and show that the source-distortion tuple remains successively refinable \cite{EquitzC_1991_Successive-refinement-of-information,Koshelev_1981_Estimation_Mean_Error} under an additional strong-sense perfect perception constraint.

\section{Problem Formulation and Definitions}
\label{sec:Problem Formulation}

\subsection*{Notation}
Random variables are in capital case (e.g., $X$) and their realizations are in lower case (e.g., $x$). We use calligraphic font (e.g., $\mathcal{X}$) to denote all sets. We use $\bbR_+$ and $\bbN$ to denote the sets of nonnegative real numbers and positive integers, respectively. Random vectors of length $n$ and their particular realizations are denoted by $X^n:= (X_1, \ldots, X_n)$ and $x^n:=(x_1,\ldots,x_n)$, respectively. All logarithms are base $2$. For any integer $a\in\bbN$, we use $[a]$ to denote $[1:a]$. The set of all probability distributions on a set $\calX$ is denoted as $\calP(\calX)$. We use $\mathbbm{1}(\cdot)$ to denote the indicator function. Finally, we follow~\cite{Gamal_2011_Network-Information-Theory} for notation of information-theoretic quantities.

\subsection{Problem Formulation}
Fix three positive integers $(n,M_1,M_2)\in\bbN^3$ and two nonnegative real numbers $(D_1,D_2)\in\mathbb R_+^2$. Fix a source distribution $P_X \in \mathcal{P}(\mathcal{X})$ defined on a finite alphabet $\mathcal{X}$. Consider a memoryless source sequence $X^n$ that is generated i.i.d. from $P_X$. Let $K$ be a common random variable taking values in an alphabet $\mathcal{K}$. As shown in Fig.~\ref{fig:successive-refinement}, in the SR problem with both distortion and strong-sense perfect perception constraints, one aims to compress the source sequence $X^n$ into a base-layer message $S_1\in [M_1]$ and a refinement-layer message $S_2\in [M_2]$ such that the source sequence $X^n$ is reconstructed as $(\hat X_1^n,\hat X_2^n)\in \mathcal{X}^n \times \mathcal{X}^n$ within distortion and strong-sense perfect perception levels $(D_1,P_X)$ and $(D_2,P_X)$, respectively. The encoding is done via an encoder $f$ and decoding is performed by two decoders $(\phi_1,\phi_2)$. We assume that the encoder and both decoders share unlimited common randomness, so no rate constraint is imposed on $K$.

In short, one aims to progressively compress a memoryless source sequence using a base layer and a refinement layer, such that the first decoder produces a coarse reconstruction and the second decoder produces a refined reconstruction in terms of distortion levels, while both reconstructions satisfy the strong-sense perfect perception constraint.

\begin{figure}[t]
\centering

\def\EncoderWidth{1.95cm}
\def\EncoderHeight{1.5cm}

\begin{tikzpicture}[
x=1cm,
y=1cm,
font=\large,
>={Latex[length=2.2mm,width=1.5mm]},
wire/.style={
line width=0.8pt
},
flow/.style={
->,
line width=0.8pt
},
crflow/.style={
->,
line width=0.8pt,
dashed
},
block/.style={
draw,
line width=0.8pt,
minimum width=1.90cm,
minimum height=0.80cm,
inner sep=0pt,
font=\Large
},
encoder/.style={
draw,
line width=0.8pt,
minimum width=\EncoderWidth,
minimum height=\EncoderHeight,
inner sep=0pt,
font=\Large
},
outerbox/.style={
draw,
line width=0.2pt,
inner sep=6.0pt
},
signal/.style={
fill=white,
inner sep=1.5pt,
font=\large
}
]

\node[encoder] (enc) at (3.35,0.10) {$f$};

\node[block] (phi1) at (8.15, 0.75) {$\phi_1$};
\node[block] (phi2) at (8.15,-0.55) {$\phi_2$};

\node[outerbox, fit=(phi1)(phi2)] (decbox) {};

\node[
block,
minimum width=2.40cm,
minimum height=0.70cm
] (rc) at (5.75,2.25)
{$K$};

\draw[crflow] (rc.south) -- (enc.north);
\draw[crflow] (rc.south) -- (decbox.north);

\coordinate (input-left) at (0.00,0.10);

\draw[flow]
(input-left)
--
node[above=2pt,pos=0.32] {$X^n\sim P_X^n$}
(enc.west);
\coordinate (encout1) at (enc.east |- phi1.west);
\coordinate (encout2) at (enc.east |- phi2.west);

\draw[flow]
(encout1)
--
node[signal,above=2pt,pos=0.48] {$S_1$}
(phi1.west);

\draw[flow]
(encout2)
--
node[signal,above=2pt,pos=0.46] {$S_2$}
(phi2.west);

\coordinate (m1branch)
at ($(phi1.west)+(-0.58,0)$);

\coordinate (phi2upperinput)
at ($(phi2.west)+(0,0.22)$);
\draw[flow]
(m1branch)
|-
(phi2upperinput);

\draw[flow]
(phi1.east)
--
++(3.00,0)
node[above=3pt,pos=0.56]
{$(\hat X_1^n\sim P_X^n,\,D_1)$};

\draw[flow]
(phi2.east)
--
++(3.00,0)
node[above=3pt,pos=0.56]
{$(\hat X_2^n\sim P_X^n,\,D_2)$};

\end{tikzpicture}

\caption{The SR randomized source coding system under strong-sense perfect perception with unlimited common randomness.}
\label{fig:successive-refinement}
\end{figure}
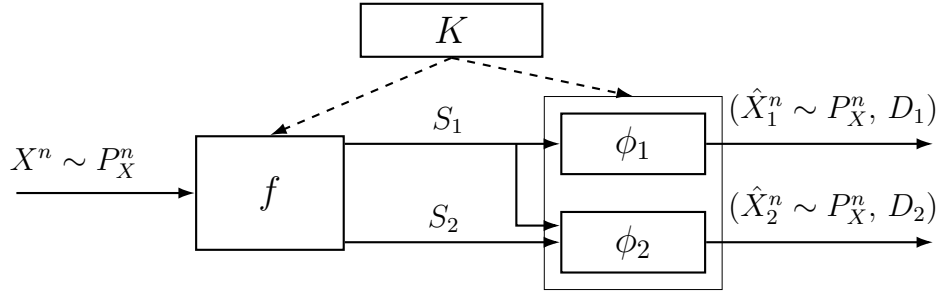
\subsection{Definitions}
A randomized SR code is formally defined as follows.

\begin{definition}
An $(n,M_1,M_2)$-code consists of one encoder
\begin{align}
f:\mathcal X^n\times\mathcal K
\to\mathcal P([M_1]\times[M_2]),
\end{align}
and two decoders
\begin{align}
\phi_1:[M_1]\times\mathcal K
&\to\mathcal P(\mathcal X^n),\\
\phi_2:[M_1]\times[M_2]\times\mathcal K
&\to\mathcal P(\mathcal X^n).
\end{align}
\end{definition}
The code definition differs from the classical SR model~\cite{Rimoldi_1994_Successive-refinement-of-information-characterization-of-the-achievable-rates} in the following two respects. On the one hand, whereas the formulation in~\cite{Rimoldi_1994_Successive-refinement-of-information-characterization-of-the-achievable-rates} employs two encoders that separately generate $S_1\in[M_1]$ and $S_2\in[M_2]$, we use an equivalent single-encoder representation that jointly generates $(S_1,S_2)$ while preserving the layered decoding structure. On the other hand, consistent with randomized coding formulations~\cite{Saldi_2015_Output_Constrained,Cuff_2013_Distributed_Channel_Synthesis}, we allow both the encoder and the decoders to be stochastic to satisfy the strong-sense perfect perception constraints.



To evaluate the distortion constraints, consider the following two bounded distortion measures: $\Delta_1:\mathcal{X}\times\mathcal{X}\to[0,\infty)$,  $\Delta_2:\mathcal{X}\times\mathcal{X}\to[0,\infty)$ such that for each $x\in\mathcal{X}$, there exists $(\hat x_1,\hat x_2)\in\mathcal{X}\times\mathcal{X}$ satisfying $\Delta_1(x,\hat x_1)=0$ and $\Delta_2(x,\hat x_2)=0$. For each $i\in[2]$, the corresponding normalized $n$-letter distortion measure is defined as
\begin{align}
\Delta_i^{(n)}(X^n,\hat X_i^n):= \frac{1}{n}\sum_{t\in[n]}\Delta_i(X_t,\hat X_{i,t}).\label{normalized-n-letter-distortion}
\end{align}
For any random variable $Z$ taking values in $\mathcal{Z}$ and any $P_Z\in\mathcal{P}(\mathcal{Z})$, we use $Z\sim P_Z$ to denote that $Z$ is distributed according to $P_Z$. Accordingly,
$Z^n\sim P_Z^n$ denotes that $Z^n$ follows the product distribution $P_Z^n$. The RDP region for SR under strong-sense perfect perception is defined as follows.
\begin{definition}\label{Definition:usp}
Given any nonnegative real numbers $(R_1,R_2)$, a rate pair $(R_1,R_2)\in\mathbb R_+^2$ is said to be achievable if there exists a sequence of $(n,M_1,M_2)$-codes such that
\begin{align}
\limsup_{n\to\infty}\frac{1}{n}\log M_1
&\le R_1,\\
\limsup_{n\to\infty}\frac{1}{n}\log(M_1M_2)
&\le R_2, \label{eq: R_2}
\end{align}
and, for each $i\in[2]$,
\begin{align}
\limsup_{n\to\infty}\bbE\!\left[
\Delta_i^{(n)}(X^n,\hat X_i^n)
\right]
&\le D_i,\label{eq:def distortion constraint}\\
\hat X_i^n
&\sim P_X^n.\label{eq:def perception constraint}
\end{align}
The convex closure of the set of all achievable rate pairs is called the RDP region and is denoted by $\calR_{\infty}(D_1,D_2|P_X)$\footnote{The subscript $\infty$ denotes unlimited common randomness shared by the encoder and two decoders.}.
\end{definition}
Consistent with the original study of Rimoldi~\cite{Rimoldi_1994_Successive-refinement-of-information-characterization-of-the-achievable-rates} and subsequent studies~\cite{Steinberg_2004_Successive_Refinement,Zhou_2017_Second_Order}, in \eqref{eq: R_2}, we use $R_2$ to denote the sum rate.
We next define the notion of successive refinability~\cite{EquitzC_1991_Successive-refinement-of-information,Koshelev_1981_Estimation_Mean_Error} under strong-sense perfect perception. To do so, we first recall the following RDP function for the P2P case. For each $i\in[2]$, given distortion level $D\in \bbR_+$, the P2P RDP function for perfect perception with distortion function $\Delta_i$ is given by~\cite[Eq. (13)]{chen_2022_on-the-rate-distortion-perception-function}
\begin{align}
R_i(D|P_X):= \inf_{\substack{P_{\hat{X}_i|X}: \,P_{\hat{X}_i}=P_X, \\ \bbE[\Delta_i(X, \hat{X}_i)] \le D}} I(X; \hat{X}_i).\label{RDP Function}
\end{align}
\begin{definition}
Given a source distribution $P_X$, two distortion measures
$\Delta_1$ and $\Delta_2$, and a distortion pair $(D_1,D_2)\in\bbR_+^2$ such that $D_1>D_2$,
the source-distortion tuple $(P_X,\Delta_1,\Delta_2)$ is said to be
$(D_1,D_2)$-successively refinable under strong-sense perfect
perception if
\begin{align}
\bigl(R_1(D_1|P_X),R_2(D_2|P_X)\bigr)
\in \mathcal{R}_{\infty}(D_1,D_2|P_X).
\end{align}
The source-distortion tuple $(P_X,\Delta_1,\Delta_2)$ is said to be
successively refinable under strong-sense perfect perception if the
above condition holds for any $(D_1,D_2)\in\bbR_+^2$ such that $D_1>D_2$.
\end{definition}
In a nutshell, successive refinability implies that layered coding incurs no additional sum rate loss relative to separately optimal P2P coding for two distortion levels.

\section{Main Result}
\label{sec:MAIN RESULTS}
Fix any $(D_1,D_2)\in\bbR_+^2$ such that $D_1>D_2$. Define the following set of distributions:
\begin{align}
\mathcal{M}(D_1,D_2)
:=\Big\{
P_{X\hat X_1\hat X_2}\in\calP(\calX^3)\colon {}
P_{\hat X_1}=P_{\hat X_2}=P_X,~\bbE[\Delta_1(X,\hat X_1)] \le D_1,~
\bbE[\Delta_2(X,\hat X_2)] \le D_2\Big\}.\label{def:M(D1,D2)}
\end{align}
Given any distribution $P_{X\hat X_1\hat X_2}\in \mathcal{M}(D_1,D_2)$, define the following set of rate pairs:
\begin{align}
\calR(P_{X\hat X_1\hat X_2})
&:=\Big\{(R_1,R_2)\in\bbR_+^2:~R_1 \ge I(X;\hat X_1),~R_2 \ge I(X;\hat X_1,\hat X_2)\Big\}.
\end{align}

\begin{theorem}\label{Theorem:1}
The RDP region for SR under strong-sense perfect perception satisfies
\begin{align}
\calR_{\infty}(D_1,D_2|P_X)=\bigcup_{P_{X\hat X_1\hat X_2}\in \mathcal{M}(D_1,D_2)}\calR(P_{X\hat X_1\hat X_2}).
\end{align}
\end{theorem}

The proof of Theorem~\ref{Theorem:1} is provided in Section \ref{sec:proof of thm1}. Theorem~\ref{Theorem:1} follows by specializing a more general result in Lemma~\ref{Theorem:2} for SR with arbitrary prescribed product output distributions and limited common randomness (cf. Section \ref{app:general-result}). In the achievability part, we introduce intermediate reconstructions whose induced distributions approximate the prescribed product distributions under the total variation (TV) distance in Lemma \ref{lem:good code}, and further apply maximal-coupling corrections to obtain the final reconstructions whose distributions match the prescribed distributions exactly. Specifically, to prove Lemma \ref{lem:good code}, we construct an auxiliary distribution that approximates the actual distribution under TV distance and satisfies the distortion and perception constraints, and show that there exists a deterministic codebook realization for which the distortion and perception bounds hold simultaneously.
The converse follows from the standard converse proof for SR~\cite{Rimoldi_1994_Successive-refinement-of-information-characterization-of-the-achievable-rates}, with modifications required to account for the output distribution constraints.


We make the following remarks. Firstly, Theorem~\ref{Theorem:1} shows that, under unlimited common randomness, strong-sense perfect perception preserves the classical two-layer rate structure of SR. The base-layer rate governs the coarse reconstruction, while the sum rate across the two layers governs the refined reconstruction. Thus, the perception constraint does not introduce new forms of compression rate constraints; rather, it restricts the admissible reconstruction distributions by requiring both reconstruction sequences to have the same block distribution as the source sequence.

Secondly, Theorem~\ref{Theorem:1} has the same rate region structure as the classical lossy SR result without the perception constraint~\cite{Rimoldi_1994_Successive-refinement-of-information-characterization-of-the-achievable-rates}. As a sanity check, when the perception constraints are removed, our region reduces to the classical SR region \cite[Theorem~1]{Rimoldi_1994_Successive-refinement-of-information-characterization-of-the-achievable-rates}. Consequently, the RDP region under strong-sense perfect perception is in general contained in the RD region. The two regions coincide if a classical optimal reconstruction distribution also satisfies the strong-sense perfect perception requirements; otherwise, enforcing strong-sense perfect perception would require higher compression rates even with unlimited common randomness.

Thirdly, we compare Theorem~\ref{Theorem:1} with the corresponding result under the weak-sense perfect perception constraint~\cite[Theorem 7]{Zhang_2025_Universal-Rate-Distortion-Perception-Representations-for-Lossy-Compression}. The weak-sense formulation constrains only the marginal distribution of each reconstruction symbol such that for each $t\in[n]$, the reconstructed source symbol $\hatX_t$ has the same distribution as the source symbol $X_t$, which is $P_X$. In contrast, our strong-sense perfect perception constraint requires the reconstructed sequence $\hatX^n$ to have the same distribution as the source sequence $X^n$, which is the product distribution $P_X^n$. It follows from our result that, with unlimited common randomness, the two formulations yield the same first-order rate region. Furthermore, we would like to emphasize that our proof differs significantly from~\cite[Theorem 7]{Zhang_2025_Universal-Rate-Distortion-Perception-Representations-for-Lossy-Compression}. Specifically, the authors of~\cite[Theorem 7]{Zhang_2025_Universal-Rate-Distortion-Perception-Representations-for-Lossy-Compression}
applied the strong functional representation lemma~\cite{Li_2018_Strong_Functional} and generated the reconstruction symbols separately, while our proof uses superposition soft covering~\cite[Lemma 4]{Goldfeld_2020_Wiretap_Channels} followed by maximal coupling~\cite[Chapter~III]{Lindvall_2002_Lectures_Coupling} to enforce the exact product distributions at both reconstruction layers without changing the first-order rates.

\section{Numerical Example}
\label{sec:Bernoulli--Hamming Example}

We now specialize Theorem~\ref{Theorem:1} to a Bernoulli source under Hamming distortion. Let $\mathcal X=\{0,1\}$ be the source and reconstruction alphabet, and let
$P_X=\mathrm{Bern}(\rho)$ be the Bernoulli distribution with parameter $\rho\in(0,0.5)$, i.e., $\Pr\{X=1\}=\rho$ and $\Pr\{X=0\}=1-\rho$.
Let $H_{\rmb}(\cdot)$ be the binary entropy function, and $\Delta_i(a,b):=\mathbbm{1}\{a\ne b\}$ be the Hamming distortion measure for any $i\in[2]$ and $(a,b)\in\calX^2$.
Fix any $(D_1,D_2)\in\bbR_+^2$.
It follows from Definition~\ref{Definition:usp} that, for any $i\in[2]$, the distortion and perception constraints \eqref{eq:def distortion constraint} and \eqref{eq:def perception constraint} specialize to
\begin{align}
\limsup_{n\to\infty}\frac{1}{n}\sum_{t\in[n]}
\Pr\{X_t\neq\hatX_{i,t}\}&\leq D_i,\quad
\hatX_i^n\sim \mathrm{Bern}(\rho)^n,\label{eq:def two hamming}
\end{align}
respectively.
With a slight abuse of notation, we reuse $\calR_\infty(D_1,D_2|P_X)$ to represent the convex closure of all rate pairs achievable under the rate constraints in Definition~\ref{Definition:usp} and the specialized constraints in \eqref{eq:def two hamming}. 

Define $D_{\max}:=2\rho(1-\rho)$. 
For any $D\in[0,D_{\max})$, define
\begin{align}
\varphi_\rho(D):=H_\rmb(\rho)-(1-\rho)H_\rmb\bigg(\frac{D}{2(1-\rho)}\bigg)-\rho H_\rmb\bigg(\frac{D}{2\rho}\bigg),\label{eq:def_phi_ro}
\end{align}
and $\varphi_\rho(D)=0$ when $D\in[D_{\max},\infty)$.
By specializing Theorem \ref{Theorem:1} to the Bernoulli source, we obtain the following result, which shows that under strong-sense perfect perception, the Bernoulli source is successively refinable under Hamming distortion.
\begin{figure}[!t]
\centering
\includegraphics[width=0.65\linewidth]{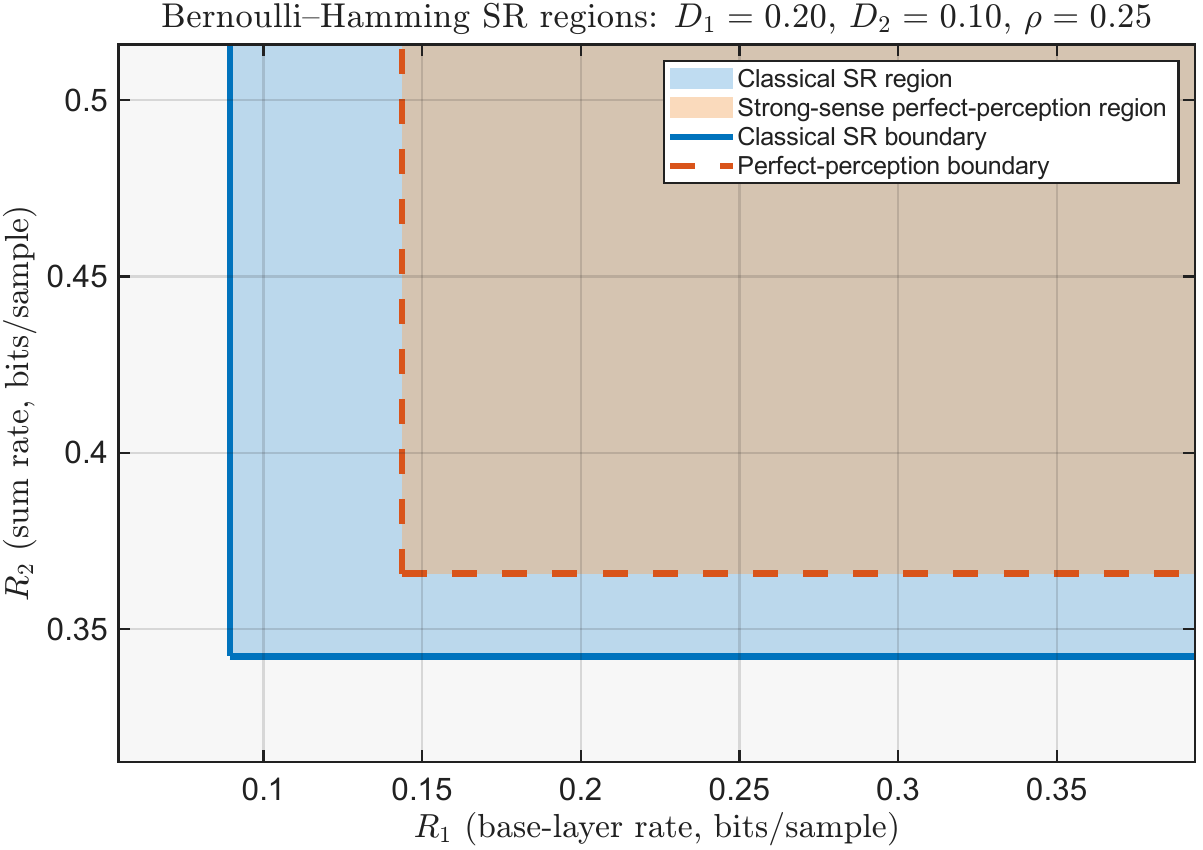}
\caption{Classical and strong-sense perfect perception SR for a Bernoulli source with parameter $\rho=0.25$ under
Hamming distortion with distortion levels $D_1=0.20$ and $D_2=0.10$.}
\label{fig:strong-sense-perfect-perception}
\end{figure}

\begin{corollary}\label{Corollary}
    For every $0\le D_2\le D_1\le D_{\max}$, the RDP region for SR of the Bernoulli source under Hamming distortion and strong-sense perfect perception is
    \begin{align}
    \begin{aligned}
    \mathcal R_\infty(D_1,D_2| P_X)=
    \big\{(R_1,R_2)\in\mathbb R_+^2:~
    R_1\ge \varphi_\rho(D_1),
    R_2\ge \varphi_\rho(D_2)\big\},
    \end{aligned}
    \label{eq:binary-sr-region}
    \end{align}
where $P_X=\mathrm{Bern}(\rho)$.
Consequently, the P2P RDP optima at the two distortion levels
can be attained simultaneously, and hence the Bernoulli source is
$(D_1,D_2)$-successively refinable under Hamming distortion and strong-sense perfect perception.
\end{corollary}
The proof of Corollary~\ref{Corollary} is available in Appendix~\ref{sec:Corollary1}.
To illustrate Corollary~\ref{Corollary}, consider $\rho=0.25$,
$D_1=0.20$, and $D_2=0.10$.
In Fig.~\ref{fig:strong-sense-perfect-perception}, we plot the rate region for SR with perfect perception in Corollary \ref{Corollary} versus the corresponding result without the perception constraint~\cite[Section~V-B]{EquitzC_1991_Successive-refinement-of-information}. Specifically, the solid and dashed lines depict the classical no-perception and strong-sense perfect-perception boundaries, respectively. As observed, the latter region is strictly contained in the former, showing that exact distribution matching incurs a rate penalty even with unlimited common randomness. Nevertheless, the perfect perception still preserves the property of successive refinability.

\section{Proof of Theorem~\ref{Theorem:1}}
\label{sec:proof of thm1}

This section presents the proof of Theorem~\ref{Theorem:1}. Specifically, Section \ref{app:general-result} analyzes a  general setting for SR with arbitrary prescribed product output distributions and limited common randomness, and specializes the general rate region to prove Theorem~\ref{Theorem:1}; Section \ref{app:achievability} analyzes the achievability part of the general results; and Section~\ref{app:converse} analyzes the converse part of the general results.

\subsection{General Results and Specialization}
\label{app:general-result}

As shown in Fig.~\ref{fig:successive-refinement2}, we consider the SR with limited common randomness, where the two reconstruction sequences are required to have arbitrary prescribed product distributions that are not necessarily the product source distribution. Fix any integer $n\in\bbN$. Let $X^n \sim P_X^n$ be a memoryless source over the finite alphabet $\calX$, and fix two prescribed reconstruction distributions $(\psi_1,\psi_2) \in \calP(\calX)^2$, three positive integers $(M_1,M_2,M_{\rm c}) \in \mathbb{N}^3$, and two distortion levels $(D_1,D_2) \in \mathbb{R}_+^2$. Let $K$ be uniformly distributed over $[M_{\rm c}]$, which is independent of $X^n$, and shared by the encoder $f$ and both decoders $(\phi_1,\phi_2)$. The encoder $f$ uses $(X^n,K)$ to generate a base-layer message $S_1 \in [M_1]$ and a refinement-layer message $S_2 \in [M_2]$. Subsequently, using $S_1$ and $K$, the decoder $\phi_1$ generates a source estimate $\hat{X}_1^n$  while the decoder $\phi_2$ generates a refined estimate $\hat{X}_2^n$ using $(S_1,S_2,K)$. For each $i \in [2]$, the reconstruction sequence $\hat{X}_i^n$ is required to be within distortion level $D_i$ from the source sequence $X^n$ and the output distribution constraint requires that the distribution of $\hatX_i^n$ is exactly $\psi_i^n$.

\begin{figure}[t]
\centering
\def\EncoderWidth{1.95cm}
\def\EncoderHeight{1.5cm}

\begin{tikzpicture}[
x=1cm,
y=1cm,
font=\large,
>={Latex[length=2.2mm,width=1.5mm]},
wire/.style={
line width=0.8pt
},
flow/.style={
->,
line width=0.8pt
},
crflow/.style={
->,
line width=0.8pt,
dashed
},
block/.style={
draw,
line width=0.8pt,
minimum width=1.90cm,
minimum height=0.80cm,
inner sep=0pt,
font=\Large
},
encoder/.style={
draw,
line width=0.8pt,
minimum width=\EncoderWidth,
minimum height=\EncoderHeight,
inner sep=0pt,
font=\Large
},
outerbox/.style={
draw,
line width=0.2pt,
inner sep=6.0pt
},
signal/.style={
fill=white,
inner sep=1.5pt,
font=\large
}
]

\node[encoder] (enc) at (3.35,0.10) {$f$};

\node[block] (phi1) at (8.15, 0.75) {$\phi_1$};
\node[block] (phi2) at (8.15,-0.55) {$\phi_2$};

\node[outerbox, fit=(phi1)(phi2)] (decbox) {};

\node[
block,
minimum width=2.40cm,
minimum height=0.70cm
] (rc) at (5.75,2.25)
{$K\in [M_{\rm{c}}]$};

\draw[crflow] (rc.south) -- (enc.north);
\draw[crflow] (rc.south) -- (decbox.north);

\coordinate (input-left) at (0.00,0.10);

\draw[flow]
(input-left)
--
node[above=2pt,pos=0.32] {$X^n\sim P_X^n$}
(enc.west);

%
\coordinate (encout1) at (enc.east |- phi1.west);
\coordinate (encout2) at (enc.east |- phi2.west);

\draw[flow]
(encout1)
--
node[signal,above=2pt,pos=0.48] {$S_1$}
(phi1.west);

\draw[flow]
(encout2)
--
node[signal,above=2pt,pos=0.46] {$S_2$}
(phi2.west);

\coordinate (m1branch)
at ($(phi1.west)+(-0.58,0)$);

\coordinate (phi2upperinput)
at ($(phi2.west)+(0,0.22)$);

\draw[flow]
(m1branch)
|-
(phi2upperinput);

\draw[flow]
(phi1.east)
--
++(3.00,0)
node[above=3pt,pos=0.56]
{$(\hat X_1^n\sim \psi_1^n,\,D_1)$};

\draw[flow]
(phi2.east)
--
++(3.00,0)
node[above=3pt,pos=0.56]
{$(\hat X_2^n\sim \psi_2^n,\,D_2)$};

\end{tikzpicture}

\caption{The SR randomized source coding system with prescribed product output distributions and limited common randomness.}
\label{fig:successive-refinement2}
\end{figure}
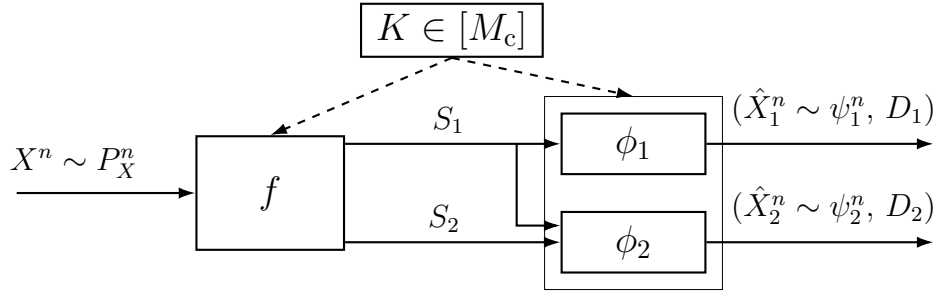

A randomized SR code is defined as follows.
\begin{definition}
An $(n,M_1,M_2,M_{\rm c})$-code consists of one encoder
\begin{align}
f:\calX^n\times[M_{\rm c}]
\to\calP([M_1]\times[M_2]),
\end{align}
and two decoders
\begin{align}
\phi_1:[M_1]\times[M_{\rm c}]
&\to\calP(\calX^n),\\*
\phi_2:[M_1]\times[M_2]\times[M_{\rm c}]
&\to\calP(\calX^n).
\end{align}
\end{definition}

To characterize the amount of common randomness, let $R_{\rm c}\in\bbR_+$ denote the common randomness rate. The corresponding rate region is defined as follows.
\begin{definition}\label{def:exact-stochastic-region}
A rate tuple $(R_1,R_2,R_{\rm c})\in\bbR_+^3$ is said to be achievable if there exists a sequence of $(n,M_1,M_2,M_{\rm c})$-codes such that
\begin{align}
\limsup_{n\to\infty}\frac{1}{n}\log M_1
&\leq R_1,\label{eq:common-R1}\\
\limsup_{n\to\infty}\frac{1}{n}\log(M_1M_2)
&\leq R_2,\label{eq:common-R2}\\
\limsup_{n\to\infty}\frac{1}{n}\log M_{\rm c}
&\leq R_{\rm c},\label{eq:common-Rc}
\end{align}
and, for each $i\in[2]$,
\begin{align}
\limsup_{n\to\infty}
\bbE\big[\Delta_i^{(n)}(X^n,\hat X_i^n)\big]
&\leq D_i,\label{eq:common-distortion}\\
\hat X_i^n&\sim\psi_i^n.
\label{eq:exact-output}
\end{align}
The convex closure of the set of all achievable rate tuples is called the optimal achievable rate region and is denoted by $\calR_{\rm SR}(D_1,D_2)$.
\end{definition}
Compared with the SR problem under strong-sense perfect perception considered in Section~\ref{sec:Problem Formulation}, the present
formulation is more general in two aspects. Firstly, the prescribed reconstruction distributions $\psi_1$ and $\psi_2$ can be arbitrary instead of being the source distribution $P_X$. Secondly, the amount of common randomness is explicitly constrained through $M_{\rm c}$, rather than being unlimited.
Fix any $(D_1,D_2)\in\bbR_+^2$ such that $D_1>D_2$. Let $U_1$ and $U_2$ be auxiliary random variables taking values in finite alphabets $\calU_1$ and $\calU_2$, respectively, such that $\hat X_1-U_1-(X,U_2)$ and $\hat X_2-(U_1,U_2)-(X,\hat X_1)$ form Markov chains. Define the following set of distributions:
\begin{align}\label{general-set}
\calM_{\rm SR}(D_1,D_2):=
\left\{
\begin{aligned}
P_{XU_1U_2\hat X_1\hat X_2}
&\in\calP(\calX^3\times\calU_1\times\calU_2):~
P_{\hat X_1}=\psi_1,~P_{\hat X_2}=\psi_2,~\bbE[\Delta_1(X,\hat X_1)]\le D_1,\\
&\qquad\qquad\bbE[\Delta_2(X,\hat X_2)]\le D_2,~
|\calU_1|\le 3|\calX|+4,~
|\calU_2|\le 2|\calX|+1
\end{aligned}
\right\}.
\end{align}
Given any distribution $P_{XU_1U_2\hat X_1\hat X_2}\in\calM_{\rm SR}(D_1,D_2)$, define the following set of rate tuples:
\begin{align}\label{general-region}
\mathcal{L}_{\rm SR}(P_{XU_1U_2\hat X_1\hat X_2})
:=
\left\{
\begin{aligned}
(R_1,R_2,R_{\rm c})\in\bbR_+^3 &:~ R_1 \ge I(X;U_1),~R_2 \ge I(X;U_1,U_2),~R_1+R_{\rm c} \ge I(\hat X_1;U_1),\\
&\qquad R_1+R_{\rm c} \ge I(\hat X_2;U_1),~R_2+R_{\rm c} \ge I(\hat X_2;U_1,U_2)
\end{aligned}
\right\}.
\end{align}

\begin{lemma}\label{Theorem:2}
The rate-distortion region for SR under prescribed product output distributions and limited common randomness in Definition~\ref{def:exact-stochastic-region} satisfies
\begin{align}
\calR_{\rm SR}(D_1,D_2)
=
\bigcup_{P_{XU_1U_2\hat X_1\hat X_2}\in\calM_{\rm SR}(D_1,D_2)}\calL_{\rm SR}(P_{XU_1U_2\hat X_1\hat X_2}).
\end{align}
\end{lemma}
The achievability and converse proofs of Lemma~\ref{Theorem:2} are provided in Sections \ref{app:achievability} and \ref{app:converse}, respectively.

In the following, we specialize the general rate region in Lemma~\ref{Theorem:2} to obtain the rate region stated in Theorem~\ref{Theorem:1}.
Setting $\psi_1=\psi_2=P_X$ and removing the constraints of the common randomness rate $R_\rmc$ from Lemma~\ref{Theorem:2} yield
\begin{align}
\calR_{\infty}(D_1,D_2|P_X)
=\bigcup_{P_{XU_1U_2\hat X_1\hat X_2}\in\calM_{\rm SR}(D_1,D_2)}
\Big\{(R_1,R_2)\in\bbR_+^2:R_1\ge I(X;U_1),~R_2\ge I(X;U_1,U_2)\Big\},
\end{align}
where $\calM_{\rm SR}(D_1,D_2)$ is defined as in \eqref{general-set} with $\psi_1=\psi_2=P_X$.
It remains to show that the auxiliary random variables $U_1$ and $U_2$ can be eliminated without changing the resulting rate region.
Recall that $\calM(D_1,D_2)$ consists of all joint distributions $P_{X\hat X_1\hat X_2}$ satisfying the prescribed marginal and distortion constraints in \eqref{def:M(D1,D2)}.
Fix any $P_{XU_1U_2\hat X_1\hat X_2}\in\calM_{\rm SR}(D_1,D_2)$.
Marginalizing over $(U_1,U_2)$ yields $P_{X\hat X_1\hat X_2}\in\calM(D_1,D_2)$.
Moreover, it follows from the Markov chains involving $U_1$ and $U_2$ defined before \eqref{general-set} that $X-U_1-\hatX_1$ and $X-(U_1,U_2)-(\hatX_1,\hatX_2)$ also form Markov chains, leading to
\begin{align}
I(X;U_1)&\ge I(X;\hat X_1),\\
I(X;U_1,U_2)&\ge I(X;\hatX_1,\hatX_2),
\end{align}
via the data-processing inequality.
Consequently, we obtain
\begin{align}
\label{eq:ach th1}
\calR_{\infty}(D_1,D_2|P_X)\subseteq
\bigcup_{P_{X\hat X_1\hat X_2}\in\calM(D_1,D_2)}
\calR(P_{X\hat X_1\hat X_2}).
\end{align}
Conversely, fix any $P_{X\hat X_1\hat X_2}\in\calM(D_1,D_2)$, $(R_1,R_2)\in\calR(P_{X\hat X_1\hat X_2})$, and set $U_1=\hat X_1$, $U_2=\hat X_2$. It follows that $P_{XU_1U_2\hat X_1\hat X_2}\in\calM_{\rm SR}(D_1,D_2)$, and
\begin{align}
I(X;U_1)&=I(X;\hat X_1),\\
I(X;U_1,U_2)&=I(X;\hatX_1,\hatX_2).
\end{align}
Consequently, we obtain $(R_1,R_2)\in\calR_{\infty}(D_1,D_2|P_X)$, which implies
\begin{align}
\label{eq:con th1}
\bigcup_{P_{X\hat X_1\hat X_2}\in\calM(D_1,D_2)}
\calR(P_{X\hat X_1\hat X_2})
\subseteq
\calR_{\infty}(D_1,D_2|P_X).
\end{align}
Combining \eqref{eq:ach th1} and \eqref{eq:con th1} completes the proof of Theorem \ref{Theorem:1}.


\subsection{Achievability Proof of Lemma~\ref{Theorem:2}}\label{app:achievability}

To satisfy the prescribed product output distribution constraints, we first
introduce intermediate reconstructions whose induced distributions approximate the prescribed product distributions under TV distance, and further apply maximal coupling corrections to obtain the final reconstructions whose distributions match the prescribed distributions exactly. For clarity, we specify our coding scheme as follows.
Recall that $\calM_{\rm SR}(\cdot,\cdot)$ denotes the set of joint distributions satisfying the prescribed constraints in~\eqref{general-set}, and $\calL_{\rm SR}(\cdot)$ denotes the corresponding set of rate tuples defined in~\eqref{general-region}.
Let $\tilX_i^n$ for any $i\in[2]$ be the intermediate reconstruction generated by the $i$-th decoder.
Recall that $(D_1,D_2)$ are distortion levels.
Fix any joint distribution $P_{XU_1U_2\tilde X_1\tilde X_2} \in\calM_{\rm SR}(D_1,D_2)$, and any rate tuple $(R_1,R_2,R_{\rm c})\in\calL_\mathrm{SR}(P_{XU_1U_2\tilde X_1\tilde X_2})$. In the following, we use the distribution $P_{XU_1U_2\tilde X_1\tilde X_2}$ to construct our coding scheme and any marginal (conditional) distribution is induced by this joint distribution.
Let $M_1:=\lfloor2^{nR_1}\rfloor$, $M_2:=\lfloor2^{n(R_2-R_1)}\rfloor$, and $M_\rmc:=\lfloor2^{nR_{\rm c}}\rfloor$.
For any $(s_1,k)\in[M_1]\times[M_{\rmc}]$, independently generate $U_1^n(s_1,k)\sim P_{U_1}^n$. Conditioned on the first-layer codebook, for each $(s_1,s_2,k)\in[M_1]\times[M_2]\times[M_{\rmc}]$, independently generate $U_2^n(s_1,s_2,k)\sim P_{U_2|U_1}^n
(\cdot|U_1^n(s_1,k))$.
For any $s_1\in[M_1]$, $s_2\in[M_2]$ and $k\in[M_\rmc]$, let 
$\bZ_n:=\big\{U_1^n(s_1,k),U_2^n(s_1,s_2,k)\big\}$ be the resulting random codebook, and let $\bz_n$ be a realization of $\bZ_n$, whose codewords are denoted by $u_1^n(s_1,k)$ and $u_2^n(s_1,s_2,k)$. Let $K$ be uniformly distributed over $[M_{\rmc}]$ and independent of $X^n\sim P_X^n$.


With the above codebooks, our coding scheme operates as follows with a joint encoder $f$ and two decoders $(\phi_1,\phi_2)$. Given $(X^n,K)=(x^n,k)$, we use the following likelihood encoder to generate $(S_1,S_2)\in [M_1]\times [M_2]$:
\begin{align}
\label{eq:encoder def}
f(s_1,s_2|x^n,k)
&:=
\frac{P_{X|U_1U_2}^n(x^n|u_1^n(s_1,k),u_2^n(s_1,s_2,k))
}{\sum_{\tils_1\in[M_1],\tils_2\in[M_2]}
P_{X|U_1U_2}^n(x^n|u_1^n(\tils_1,k),u_2^n(\tils_1,\tils_2,k))
}.
\end{align}
If the denominator is zero, we define $f(1,1|x^n,k)=1$. 
Upon observing $(S_1,K)=(s_1,k)$, the decoder $\phi_1$ generates $\tilX_1^n\in\calX^n$ according to
\begin{align}
\phi_1(\tilx_1^n|s_1,k):= P^n_{\tilX_1|U_1}(\tilx_1^n|u_1^n(s_1,k)).\label{F_1}
\end{align}
Similarly, upon observing $(S_1,S_2,K)=(s_1,s_2,k)$, the decoder $\phi_2$ generates $\tilX_2^n\in\calX^n$ according to
\begin{align}
\phi_2(\tilde{x}_2^n|s_1,s_2,k):=  P^n_{\tilX_2|U_1U_2}(\tilde{x}_2^n|u_1^n(s_1,k),u_2^n(s_1,s_2,k)).\label{F_2}
\end{align}
We next construct the final reconstructions satisfying the exact output distribution constraints, i.e., $\hatX_i^n\sim\psi_i^n$ for any $i\in[2]$. 
For any $i\in[2]$, let
\begin{align}
\label{eq:def of T i}
T_i^{(n)}:\calX^n\to\calP(\calX^n),
\end{align}
specified later. The modified first- and second-layer decoders are defined as
\begin{align}
\hat{\phi}_1(\hatx_1^n|s_1,k)&:=
\sum_{\tilx_1^n\in\calX^n}
\phi_1(\tilx_1^n|s_1,k)
T_1^{(n)}(\hatx_1^n|\tilx_1^n),\label{eq:ach-modified-decoder1}\\
\hat{\phi}_2(\hatx_2^n|s_1,s_2,k)&:=
\sum_{\tilx_2^n\in\calX^n}
\phi_2(\tilx_2^n|s_1,s_2,k)
T_2^{(n)}(\hatx_2^n|\tilx_2^n),
\label{eq:ach-modified-decoder2}
\end{align}
respectively.

In the following, Section~\ref{subapp:aux tilX} shows that there exists a deterministic codebook for which the intermediate reconstructions $\{\tilX_i^n\}_{i\in[2]}$ satisfy the distortion constraints and whose distributions approximate the prescribed product distributions, whereas Section~\ref{subapp:final hatX} shows that the final reconstructions $\{\hatX_i^n\}_{i\in[2]}$ satisfy the exact output distribution constraints with only a vanishing additional distortion to complete the proof.

\subsubsection{Existence of a Good Deterministic Codebook}
\label{subapp:aux tilX}

For any $n\in\bbN$ and any realization $\bz_n$ of the random codebook $\bZ_n$, let $P^{\bz_n}_{X^nKS_1S_2\tilX_1^n\tilX_2^n}$ denote the joint distribution induced by $\bz_n$.\footnote{For clarity, we use $P^{\bZ_n}$ to denote the random probability distribution induced by the random codebook $\bZ_n$, and $P^{\bz_n}$ to denote the corresponding probability distribution for a fixed realization $\bZ_n=\bz_n$. The same convention applies to all other distributions.
}
For any $i\in[2]$, let $P_{\tilX_i^n}^{\bz_n}$ be the corresponding marginal distribution of $\tilX_i^n$. Recall that $\psi_i\in\calP(\calX)$ is the prescribed single-letter output distribution at layer $i\in[2]$, and its $n$-letter product distribution is given by $\psi_i^n(x_i^n)=\prod_{t\in[n]}\psi_i(x_{i,t})$, for any $x_i^n\in\calX^n$.
Given any two probability distributions $(P,Q)\in\mathcal{P}(\mathcal{X})^2$, we use $\|P-Q\|_{\mathrm{TV}}:=\frac{1}{2}\sum_{x\in\mathcal{X}}|P(x)-Q(x)|$ to denote the TV distance.

\begin{lemma}
\label{lem:good code}
There exists a sequence of deterministic codebooks $\{\bz_n\}_{n\in\bbN}$ such that, for each $i\in[2]$,
\begin{align}
\limsup_{n\to\infty}\bbE\big[\Delta_i^{(n)}(X^n,\tilX_i^n)\big]&\le D_i,
\label{eq:good code distortion}\\
\lim_{n\to\infty}\big\|P_{\tilX_i^n}^{\bz_n}-\psi_i^n\big\|_{\rm TV}&=0.
\label{eq:good code perception}
\end{align}
\end{lemma}

The proof of Lemma~\ref{lem:good code} is provided in Appendix~\ref{app:lem good code}. Specifically, we first relate the actual distribution to an auxiliary distribution, and show that the actual distribution satisfies the distortion and perception constraints via the auxiliary distribution on average over the random codebook. Furthermore, we show that there exists a deterministic codebook for which the distortion and perception constraints are satisfied simultaneously.

\subsubsection{Exact Output Distribution and Final Steps}
\label{subapp:final hatX}

Let $\Pi(P,Q)$ be the set of all couplings of distributions $P$ and $Q$. It follows from the maximal coupling theorem \cite[Proposition 4.7]{Levin_2017_Markov_Chains_and_Mixing_Times} that there exists a coupling $\Gamma_i^{(n)}\in\Pi\big(P_{\tilX_i^n}^{\bz_n},\psi_i^n\big)$ for any $i\in[2]$ such that
\begin{align}
\Pr_{\Gamma_i^{(n)}}
\big\{\tilX_i^n\neq\hatX_i^n\big\}&=
\big\|P_{\tilX_i^n}^{\bz_n}-\psi_i^n
\big\|_{\rm TV}.
\label{eq:coupling bound}
\end{align}
For any $i\in[2]$ and $\hatx_i^n\in\calX^n$, it follows from the definition of $\Pi(\cdot,\cdot)$ that the marginal distribution of $\Gamma_i^{(n)}$ satisfies
\begin{align}
\sum_{\tilx_i^n\in\calX^n}\Gamma_i^{(n)}(\tilx_i^n,\hatx_i^n)
=\psi_i^n(\hatx_i^n).\label{eq:marginal result}
\end{align}
For any $P_{\tilX_i^n}^{\bz_n}(\tilx_i^n)>0$, $\hatx_i^n\in\calX^n$ and $i\in[2]$, let
\begin{align}
T_i^{(n)}(\hatx_i^n|\tilx_i^n):=\frac{\Gamma_i^{(n)}(\tilx_i^n,\hatx_i^n)}{P_{\tilX_i^n}^{\bz_n}(\tilx_i^n)}\label{eq:redef Ti},
\end{align}
and $T_i^{(n)}(\cdot|\tilx_i^n)$ may be chosen arbitrarily when $P_{\tilX_i^n}^{\bz_n}(\tilx_i^n)=0$.
Consequently, for any $i\in[2]$ and $\hatx_i^n\in\calX^n$, we obtain
\begin{align}
P_{\hatX_i^n}(\hatx_i^n)
&=\sum_{\tilx_i^n\in\calX^n}P_{\tilX_i^n}^{\bz_n}(\tilx_i^n)T_i^{(n)}(\hatx_i^n|\tilx_i^n)\label{eq:sumTi -1}\\
&=\sum_{\tilx_i^n\in\calX^n}\Gamma_i^{(n)}(\tilx_i^n,\hatx_i^n)\label{eq:sumTi -2}\\
&=\psi_i^n(\hatx_i^n),\label{eq:sumTi -3}
\end{align}
where \eqref{eq:sumTi -1} follows from the definition of the modified decoder in \eqref{eq:ach-modified-decoder1} and \eqref{eq:ach-modified-decoder2} and the law of total probability, \eqref{eq:sumTi -2} follows from \eqref{eq:redef Ti}, and \eqref{eq:sumTi -3} follows from \eqref{eq:marginal result}.
Recall that we use $\mathbbm{1}\{\cdot\}$ to denote the indicator function. Let $\Delta_{\max}:=\max_{i\in[2]}\max_{(x,\hatx)\in\calX^2}\Delta_i(x,\hatx)$.
For any $(x^n,\tilde x_i^n,\hat x_i^n)\in(\calX^n)^3$, it follows from \eqref{normalized-n-letter-distortion} that
\begin{align}
\Delta_i^{(n)}(x^n,\hat x_i^n)&\le\Delta_i^{(n)}(x^n,\tilde x_i^n)+\Delta_{\max}
\mathbbm{1} \{\hat x_i^n\neq \tilde x_i^n\}. \label{eq:pointwise-dist-bound}
\end{align}
As $n\to\infty$, taking expectations on \eqref{eq:pointwise-dist-bound} leads to
\begin{align}
\bbE\big[\Delta_i^{(n)}(X^n,\hat X_i^n)\big]
&\le\bbE\big[\Delta_i^{(n)}(X^n,\tilde X_i^n)\big]+\Delta_{\max}\Pr\{\hat X_i^n\neq \tilde X_i^n\}\\*
&\le D_i,
\label{eq:dist-after-coupling_2}
\end{align}
where \eqref{eq:dist-after-coupling_2} follows from Lemma \ref{lem:good code} and \eqref{eq:coupling bound} and the fact that $\Delta_{\max}$ is bounded.
Combining Lemma \ref{lem:good code}, \eqref{eq:sumTi -3} and \eqref{eq:dist-after-coupling_2} completes the achievability proof of Lemma~\ref{Theorem:2}.

\subsection{Converse Proof of Lemma~\ref{Theorem:2}}\label{app:converse}
Fix any achievable rate tuple $(R_1,R_2,R_{\rm c})$ in Definition~\ref{def:exact-stochastic-region}, we shall show that there exists a joint distribution $P_{XU_1U_2\hat X_1\hat X_2}\in\calM_{\rm SR}(D_1,D_2)$ such that
\begin{align}
(R_1,R_2,R_{\rm c})\in\calL_{\rm SR}(P_{XU_1U_2\hat X_1\hat X_2}).
\end{align}
In the following, we first derive the five required single-letter rate constraints of $\calL_{\rm SR}(P_{XU_1U_2\hat X_1\hat X_2})$ in \eqref{general-region}, and then verify that the resulting single-letter joint distribution belongs to $\calM_{\rm SR}(D_1,D_2)$.
Let $S_1 \in [M_1]$ and $S_2 \in [M_2]$ be the first- and refinement-layer messages for the source sequence $X^n$, respectively. Let $K\in[M_\rmc]$ be the common randomness, where $K$ is independent of $X^n$.
As $n\to\infty$, for any achievable rate $R_1$, it follows from \eqref{eq:common-R1} that
\begin{align}
R_1
&\geq \frac{1}{n}\log M_1\label{eq:conv-r1-a}\\
&\geq \frac{1}{n}I(X^n;S_1|K)\label{eq:conv-r1-d}\\
&= \frac{1}{n}I(X^n;S_1,K)\label{eq:conv-r1-e}\\
&= \frac{1}{n}\sum_{t\in[n]}I(X_t;S_1,K|X^{t-1})\label{eq:conv-r1-f}\\
&\geq \frac{1}{n}\sum_{t\in[n]}I(X_t;S_1,K),\label{eq:conv-r1-g}
\end{align}
where \eqref{eq:conv-r1-d} follows from the fact that $\log M_1\ge H(S_1|K)\ge I(X^n;S_1|K)$, \eqref{eq:conv-r1-e} follows from the independence of $K$ and $X^n$, \eqref{eq:conv-r1-f} follows from the chain rule for mutual
information, and \eqref{eq:conv-r1-g} follows from $X^n\sim P_X^n$ and the fact that conditioning reduces entropy.
Similarly, we obtain
\begin{align}
R_2\geq\frac{1}{n}\sum_{t\in[n]}I(X_t;S_1,S_2,K).
\label{eq:conv-r2}
\end{align}
It follows from \eqref{eq:common-R1} and \eqref{eq:common-Rc} that
\begin{align}
R_1+R_{\rm c}
&\geq\frac{1}{n}\log(M_1M_{\rm c})\label{eq:conv-rc1-a}\\*
&\geq \frac{1}{n}I(\hat X_1^n;S_1,K)\label{eq:conv-rc1-c}\\*
&=\frac{1}{n}\sum_{t\in[n]}I(\hat X_{1,t};S_1,K|\hat X_1^{t-1})\label{eq:conv-rc1-d}\\*
&\geq\frac{1}{n}\sum_{t\in[n]}I(\hat X_{1,t};S_1,K),\label{eq:conv-rc1-e}
\end{align}
where \eqref{eq:conv-rc1-c} follows from the fact that $\log(M_1M_{\rm c})\ge H(S_1,K)\ge I(\hat X_1^n;S_1,K)$, \eqref{eq:conv-rc1-d} follows from the chain rule for mutual information, and \eqref{eq:conv-rc1-e} follows from $\hat X_1^n\sim\psi_1^n$ and the fact that conditioning reduces entropy.
Similarly, we obtain
\begin{align}
R_1+R_{\rm c}&\geq\frac{1}{n}\sum_{t\in[n]}I(\hat X_{2,t};S_1,K),
\label{eq:conv-rc2}\\
R_2+R_{\rm c}&\geq\frac{1}{n}\sum_{t\in[n]}I(\hat X_{2,t};S_1,S_2,K).
\label{eq:conv-sum}
\end{align}
Let $J$ be uniformly distributed over $[n]$ and independent of all other random variables, and let $U_1:=(S_1,K,J)$, $U_2:=S_2$.
It follows from \eqref{eq:conv-r1-g} that
\begin{align}
R_1
&\ge I(X_J;S_1,K|J)\label{eq:conv-ts-r1-a}\\*
&=I(X_J;S_1,K,J)\label{eq:conv-ts-r1-b}\\*
&=I(X_J;U_1),\label{eq:conv-ts-r1-c}
\end{align}
where \eqref{eq:conv-ts-r1-a} follows from the uniformity and independence of $J$, \eqref{eq:conv-ts-r1-b} follows since $X_J$ is independent of $J$, and \eqref{eq:conv-ts-r1-c} follows from the fact that $U_1=(S_1,K,J)$. Similarly, we obtain
\begin{align}
R_2&\geq I(X_J;U_1,U_2),\label{eq:conv-sl-r2}\\*
R_1+R_{\rm c}&\geq I(\hat X_{1,J};U_1),\label{eq:conv-sl-rc1}\\*
R_1+R_{\rm c}&\geq I(\hat X_{2,J};U_1),\label{eq:conv-sl-rc2}\\*
R_2+R_{\rm c}&\geq I(\hat X_{2,J};U_1,U_2).\label{eq:conv-sl-sum}
\end{align}

We next verify that the single-letter joint distribution
$P_{X_JU_1U_2\hat X_{1,J}\hat X_{2,J}}$ belongs to
$\calM_{\rm SR}(D_1,D_2)$.
Since $X_t\sim P_X$ and $\hatX_{i,t}\sim\psi_i$ for any $t\in[n]$ and $i\in[2]$, the uniformity of $J$ yields
\begin{align}
X_J&\sim P_X,\label{eq:conv-source-marginal}\\
\hatX_{i,J}&\sim\psi_i.\label{eq:conv-source-marginal-hat}
\end{align}
Moreover, for any $i\in[2]$ and any achievable code defined in Definition~\ref{def:exact-stochastic-region}, it follows that
\begin{align}
\bbE[\Delta_i(X_J,\hat X_{i,J})]
&=
\frac{1}{n}
\sum_{t\in[n]}
\bbE[\Delta_i(X_t,\hat X_{i,t})]
\label{eq:conv-dist-sl-b}\\*
&=
\bbE[\Delta_i^{(n)}(X^n,\hat X_i^n)]
\label{eq:conv-dist-sl-c}\\*
&\leq
D_i,
\label{eq:conv-dist-sl-d}
\end{align}
where \eqref{eq:conv-dist-sl-b} follows from the
uniformity and independence of $J$, \eqref{eq:conv-dist-sl-c}
follows from \eqref{normalized-n-letter-distortion}, and
\eqref{eq:conv-dist-sl-d} follows from \eqref{eq:common-distortion}.
Recall that the first decoder uses $(S_1,K)$, whereas the second decoder uses $(S_1,S_2,K)$. Since the two reconstruction sequences are conditionally independent given their respective decoder inputs, we obtain
\begin{align}
P_{\hat X_1^n\hat X_2^n| X^n,S_1,S_2,K} = P_{\hat X_1^n|S_1,K} P_{\hat X_2^n| S_1,S_2,K}. 
\end{align}
Selecting the $J$-th coordinates and marginalizing over the remaining coordinates yield
\begin{align}
P_{\hat X_{1,J}\hat X_{2,J}
|X_J,S_1,S_2,K,J}
=P_{\hat X_{1,J}|S_1,K,J}
P_{\hat X_{2,J}|S_1,S_2,K,J}.
\label{eq:conv-decoder-sl-factorization}
\end{align}
Recall that $U_1=(S_1,K,J)$ and $U_2=S_2$. It follows that
\begin{align}
P_{\hat X_{1,J}\hat X_{2,J}|X_J,U_1,U_2}
=
P_{\hat X_{1,J}|U_1}
P_{\hat X_{2,J}|U_1,U_2},
\label{eq:conv-markov-factorization}
\end{align}
implying that $U_1$ and $U_2$ satisfy the Markov conditions $\hatX_{1,J}-U_1-(X_J,U_2)$ and $\hatX_{2,J}-(U_1,U_2)-(X_J,\hat X_{1,J})$. 
Consistent with \cite[Section 15.8]{Thomas_2006_Elements-of-information-theory}, the converse proof of Lemma~\ref{Theorem:2} is completed by combining the independence of $X_J$ and $J$ with \eqref{eq:conv-ts-r1-c}, \eqref{eq:conv-sl-r2}--\eqref{eq:conv-source-marginal-hat}, \eqref{eq:conv-dist-sl-d} and \eqref{eq:conv-markov-factorization}.

\section{Conclusion}
\label{sec:Conclusion}
We revisited the SR problem and characterized its RDP region under strong-sense perfect perception and unlimited common randomness. Our result extends P2P RDP theory to layered source coding and shows that the classical SR rate structure is preserved, while the admissible reconstruction distributions are constrained to match the source distribution. For the Bernoulli source under Hamming distortion, we further obtained a closed-form region and showed that the P2P RDP optima at two distortion levels can be achieved simultaneously. Hence, the Bernoulli source remains successively refinable under strong-sense perfect perception. Future work includes generalizations to continuous alphabet sources, finite-blocklength analysis, and multi-stage SR systems.

\appendix




\subsection{Proof of Corollary~\ref{Corollary}}
\label{sec:Corollary1}

Let $(X,\hatX)\in\calX^2$ denote the source and reconstruction random variables for the Bernoulli source under Hamming distortion, respectively. For any $(x,y)\in\{0,1\}^2$, recall that $\Delta(x,y)=\mathbbm{1}\{x\neq y\}$ is the Hamming distortion measure. It follows from~\cite[Example 1]{chen_2022_on-the-rate-distortion-perception-function} that $\varphi_\rho(D)$ defined in~\eqref{eq:def_phi_ro} equals the minimum of $I(X;\hat X)$ over all $P_{X\hat X}\in\calP(\{0,1\}^2)$ satisfying $P_X=P_{\hat X}=\mathrm{Bern}(\rho)$ and $\bbE[\Delta(X,\hat X)]\le D$, i.e.,
\begin{align}
\varphi_\rho(D)=\min_{P_{X\hat X}} I(X;\hat X).
\label{eq:binary-variational}
\end{align}
For any $d\in[0,D_{\max}]$, define a binary channel $T(d):\{0,1\}\to\calP(\{0,1\})$ with transition matrix
\begin{align}
\label{eq:def of bT}
\bT(d):=
\begin{pmatrix}
T_d(0|0) & T_d(1|0)\\
T_d(0|1) & T_d(1|1)
\end{pmatrix}=
\begin{pmatrix}
1-\frac{d}{2(1-\rho)}&
\frac{d}{2(1-\rho)}\\
\frac{d}{2\rho}&
1-\frac{d}{2\rho}
\end{pmatrix}.
\end{align}
For any $D_2\in[0,D_{\max})$ and $D_1\in[D_2,D_{\max}]$, let
\begin{align}
\delta:=\frac{D_1-D_2}{1-D_2/D_{\max}},
\label{eq:binary-delta}
\end{align}
and $\delta=0$ when $D_2=D_{\max}$.
Furthermore, define a joint distribution $\tilP_{X\hatX_1\hatX_2}$ with
\begin{align}
\tilP_{X\hatX_1\hatX_2}(x,\hatx_1,\hatx_2) := P_X(x)T_{D_2}(\hatx_2|x) T_\delta(\hatx_1|\hatx_2),
\label{eq:binary-joint-construction}
\end{align}
for any $(x,\hat x_1,\hat x_2)\in\{0,1\}^3$.
We shall show in Appendix~\ref{subapp:verify tilP} that
$\tilP_{X\hat X_1\hat X_2}\in\calM(D_1,D_2)$ and that, under $\tilP_{X\hat X_1\hat X_2}$
\begin{align}
I(X;\hat X_1)
&=\varphi_\rho(D_1),\\
I(X;\hat X_1,\hat X_2)
&=\varphi_\rho(D_2).
\end{align}
Consequently, it follows from Theorem~\ref{Theorem:1} that
\begin{align}
\Big\{(R_1,R_2)\in\bbR_+^2:
R_1\ge\varphi_\rho(D_1),~
R_2\ge\varphi_\rho(D_2)
\Big\}
\subseteq
\calR_\infty(D_1,D_2|P_X).
\label{eq:binary-achievability-inclusion}
\end{align}

Recall from \eqref{def:M(D1,D2)} that $\calM(D_1,D_2)$ denotes the set of joint distributions satisfying the distortion and strong-sense perfect perception constraints. For any $P_{X\hat X_1\hat X_2}\in\calM(D_1,D_2)$, the marginal distribution $P_{X\hatX_i}$ satisfies $P_X=P_{\hat X_i}=\operatorname{Bern}(\rho)$ and $\bbE[\Delta(X,\hat X_i)]\le D_i$ for any $i\in[2]$. It follows from \eqref{eq:binary-variational} that for any $i\in[2]$,
\begin{align}
I(X;\hatX_i)\ge\varphi_\rho(D_i).
\end{align}
Furthermore, we obtain
\begin{align}
I(X;\hatX_1,\hat X_2)\ge I(X;\hatX_2)\ge\varphi_\rho(D_2).
\end{align}
Consequently, it follows from Theorem~\ref{Theorem:1} that
\begin{align}
\calR_\infty(D_1,D_2|P_X)\subseteq\Big\{(R_1,R_2)\in\bbR_+^2:
R_1\ge\varphi_\rho(D_1),~R_2\ge\varphi_\rho(D_2)\Big\}.
\label{eq:binary-converse-inclusion}
\end{align}
Combining \eqref{eq:binary-achievability-inclusion} and \eqref{eq:binary-converse-inclusion} completes the proof of Corollary \ref{Corollary}.

\subsubsection{Verification of $\tilP_{X\hat X_1\hat X_2}$}
\label{subapp:verify tilP}
We first verify that $\tilP_{X\hat X_1\hat X_2}$ satisfies the perception and distortion constraints in \eqref{def:M(D1,D2)}. The definitions of $\bT(\cdot)$ in \eqref{eq:def of bT} and $\delta$ in \eqref{eq:binary-delta} and the matrix multiplication lead to
\begin{align}
\label{eq:T delta d1d2}
\bT(D_2)\bT(\delta)=\bT\bigg(D_2+\delta\Big(1-\frac{D_2}{D_{\max}}\Big)\bigg)=\bT(D_1).
\end{align}
Recall that $X\sim\mathrm{Bern(\rho)}$. It follows from \eqref{eq:binary-joint-construction} and \eqref{eq:T delta d1d2} that the transition matrix from $X$ to $\hat X_i$ is $\bT(D_i)$ for any $i\in[2]$. Let $\tilP_{\hat X_i}$ denote the marginal distribution induced by $\tilP_{X\hat X_1\hat X_2}$ for any $i\in[2]$. For any $i\in[2]$, it follows that
\begin{align}
\tilP_{\hatX_i}(1)
&=(1-\rho)T_{D_i}(1|0)+\rho T_{D_i}(1|1)\\*
&=(1-\rho)\frac{D_i}{2(1-\rho)}+\rho\Big(1-\frac{D_i}{2\rho}\Big)\\*
&=\rho,\label{eq:hatX bern}
\end{align}
implying that $\tilP_{\hatX_i}=P_X=\mathrm{Bern}(\rho)$. Furthermore, for any $i\in[2]$, it follows that
\begin{align}
\bbE[\Delta(X,\hatX_i)]
&=(1-\rho)T_{D_i}(1|0)+\rho T_{D_i}(0|1)\\*
&=(1-\rho)\frac{D_i}{2(1-\rho)}+\rho\frac{D_i}{2\rho}\\*
&=D_i.\label{eq:distortion tilP}
\end{align}
Consequently, we obtain $\tilP_{X\hatX_1\hat X_2}\in\calM(D_1,D_2)$.
Under $\tilP_{X\hat X_1\hat X_2}$, it follows that
\begin{align}
I(X;\hat X_i)
&=H(\hat X_i)-H(\hat X_i|X)\label{eq:I cal -1}\\
&=H_{\rmb}(\rho)-(1-\rho)H_{\rmb}\Big(\frac{D_i}{2(1-\rho)}\Big)-\rho H_{\rmb}\Big(\frac{D_i}{2\rho}\Big)\label{eq:I cal -2}\\
&=\varphi_\rho(D_i),
\label{eq:I cal -3}
\end{align}
where \eqref{eq:I cal -2} follows from the fact that $\hatX_i\sim\mathrm{Bern}(\rho)$ and $\bT(D_i)$ is the transition matrix for any $i\in[2]$, and \eqref{eq:I cal -3} follows from \eqref{eq:def_phi_ro}. It follows from \eqref{eq:binary-joint-construction} that $X-\hatX_2-\hatX_1$ forms a Markov chain. Consequently, we obtain
\begin{align}
I(X;\hat X_1,\hat X_2)
&=I(X;\hat X_2)+I(X;\hat X_1|\hat X_2)\\*
&=I(X;\hat X_2)\label{eq:II cal-2}\\*
&=\varphi_\rho(D_2),\label{eq:II cal-3}
\end{align}
where \eqref{eq:II cal-2} follows from the fact that $I(X;\hat X_1|\hat X_2)=0$, and \eqref{eq:II cal-3} follows from \eqref{eq:def_phi_ro}.
Combining \eqref{eq:hatX bern}, \eqref{eq:distortion tilP}, \eqref{eq:I cal -3} and \eqref{eq:II cal-3} completes the proof.

\subsection{Proof of Lemma \ref{lem:good code}}
\label{app:lem good code}
The proof consists of the following three steps: \emph{1)} construct an auxiliary distribution and relate the desired distortion and perception properties under the actual distribution to the auxiliary distribution; \emph{2)} show that the reconstructions $\{\tilX_i^n\}_{i\in[2]}$ satisfy the distortion constraint and whose distributions approximate the prescribed distributions $\{\psi_i^n\}_{i\in[2]}$ under the TV distance on average over the random codebook; and \emph{3)} show that there exists a deterministic codebook realization for which the distortion and perception bounds hold simultaneously.

\subsubsection{Relating the Actual and Auxiliary distributions}

Fix any $n\in\bbN$. To facilitate the analysis of the actual distribution $P^{\bz_n}_{X^nKS_1S_2\tilX_1^n\tilX_2^n}$, we introduce an auxiliary distribution, which uses the same codebook and decoders $\{\phi_i\}_{i\in[2]}$ defined in Section \ref{app:achievability}, but chooses $(S_1,S_2,K)$ uniformly and generates $X^n$ through the reverse test channel $P_{X|U_1U_2}^n$. For any realization $\bz_n$ of the random codebook $\bZ_n$ and any $(x^n,\tilx_1^n,\tilx_2^n,s_1,s_2,k) \in(\calX^{n})^3\times[M_1]\times[M_2]\times[M_{\rm c}]$, define the auxiliary distribution $Q^{\bz_n}_{X^n\tilX_1^n\tilX_2^nS_1S_2K}$ as
\begin{align}
&Q^{\bz_n}_{X^n\tilX_1^n\tilX_2^nS_1S_2K}
(x^n,\tilx_1^n,\tilx_2^n,s_1,s_2,k)
\nn\\
&\qquad:=
\frac{1}{M_1M_2M_{\rm c}}
P_{X|U_1U_2}^n
\big(
x^n\big|
u_1^n(s_1,k),
u_2^n(s_1,s_2,k)
\big)
\phi_1(\tilx_1^n|s_1,k)
\phi_2(\tilx_2^n|s_1,s_2,k).
\label{eq:auxiliary-distribution}
\end{align}
Recall from Section \ref{app:achievability} that, for any fixed codebook realization $\bz_n$, the actual distribution is given by
\begin{align}
P^{\bz_n}_{X^n\tilX_1^n\tilX_2^nS_1S_2K}
(x^n,\tilx_1^n,\tilx_2^n,s_1,s_2,k)
=\frac{1}{M_{\rm c}}P_X^n(x^n)
f(s_1,s_2|x^n,k)
\phi_1(\tilx_1^n|s_1,k)
\phi_2(\tilx_2^n|s_1,s_2,k).
\label{eq:actual-distribution}
\end{align}
For any fixed codebook realization $\bz_n$ and any $i\in[2]$, define
\begin{align}
T(\bz_n)&:=
\big\|Q^{\bz_n}_{X^n\tilX_1^n\tilX_2^nS_1S_2K}
-P^{\bz_n}_{X^n\tilX_1^n\tilX_2^nS_1S_2K}
\big\|_{\rm TV},
\label{eq:T(C)-def}\\
V_i^Q(\bz_n)
&:=\big\|Q_{\tilX_i^n}^{\bz_n}-\psi_i^n\big\|_{\rm TV},
\label{eq:Vi-def}\\
D_i^Q(\bz_n)
&:=\bbE_{Q^{\bz_n}_{X^n\tilX_i^n}}
\big[\Delta_i^{(n)}(X^n,\tilX_i^n)\big].
\label{eq:DQ-def}
\end{align}
When the codebook is random, $T(\bZ_n)$, $V_i^Q(\bZ_n)$, and $D_i^Q(\bZ_n)$ denote the corresponding random variables. Recall that $\Delta_{\max}=\max_{i\in[2]}\max_{(x,\hatx)\in\calX^2}\Delta_i(x,\hatx)$. It follows that, for any fixed codebook realization $\bz_n$ and $i\in[2]$,
\begin{align}
\bbE_{P_{X^n\tilX_i^n}^{\bz_n}}\big[\Delta_i^{(n)}(X^n,\tilX_i^n)\big]
&\le \bbE_{Q_{X^n\tilX_i^n}^{\bz_n}}\big[\Delta_i^{(n)}(X^n,\tilX_i^n)\big]+\Delta_{\max}\big\|P_{X^n\tilX_i^n}^{\bz_n}-Q_{X^n\tilX_i^n}^{\bz_n}\big\|_\mathrm{TV}
\label{eq:distortion p q actual-1}\\
&\le \bbE_{Q_{X^n\tilX_i^n}^{\bz_n}}\big[\Delta_i^{(n)}(X^n,\tilX_i^n)\big]+\Delta_{\max}\big\|Q_{X^nKS_1S_2\tilX_1^n\tilX_2^n}^{\bz_n}-P_{X^nKS_1S_2\tilX_1^n\tilX_2^n}^{\bz_n}\big\|_{\rm{TV}}\label{eq:distortion p q actual-2}\\
&= D_i^Q(\bz_n)+\Delta_{\max}T(\bz_n),
\label{eq:distortion p q actual-4}
\end{align}
where \eqref{eq:distortion p q actual-1} follows from the property of TV distance, \eqref{eq:distortion p q actual-2} follows from the contraction of TV distance under marginalization, \eqref{eq:distortion p q actual-4} follows from \eqref{eq:T(C)-def} and \eqref{eq:DQ-def}. Similarly, it follows that, for any fixed codebook realization $\bz_n$ and $i\in[2]$,
\begin{align}
\big\|P_{\tilX_i^n}^{\bz_n}-\psi_i^n\big\|_\mathrm{TV}
&\le \big\|P_{\tilX_i^n}^{\bz_n}-Q_{\tilX_i^n}^{\bz_n}\big\|_\mathrm{TV}+\big\|Q_{\tilX_i^n}^{\bz_n}-\psi_i^n\big\|_\mathrm{TV}\label{eq:perception tv P-1}\\*
&\le \big\|Q_{X^nKS_1S_2\tilX_1^n\tilX_2^n}^{\bz_n}-P_{X^nKS_1S_2\tilX_1^n\tilX_2^n}^{\bz_n}\big\|_{\rm{TV}}+\big\|Q_{\tilX_i^n}^{\bz_n}-\psi_i^n\big\|_\mathrm{TV}\label{eq:perception tv P-2}\\
&=T({\bz_n})+V_i^Q({\bz_n}),
\label{eq:perception tv P-3}
\end{align}
where \eqref{eq:perception tv P-1} follows from the triangle inequality, \eqref{eq:perception tv P-2} follows from the contraction of TV distance under marginalization, and \eqref{eq:perception tv P-3} follows from \eqref{eq:T(C)-def} and \eqref{eq:Vi-def}.
Based on \eqref{eq:distortion p q actual-4} and \eqref{eq:perception tv P-3}, we shall show in Appendix \ref{subapp:dis and per} that for any $i\in[2]$, $\lim_{n\to\infty}\bbE_{\bZ_n}[T(\bZ_n)]=0$, $\bbE_{\bZ_n}[D_i^Q(\bZ_n)]\le D_i$, and $\lim_{n\to\infty}\bbE_{\bZ_n}[V_i^Q(\bZ_n)]=0$, to verify that the actual distribution satisfies the distortion and perception constraints on average over the random codebook.
Consequently, we obtain that, for any $i\in[2]$,
\begin{align}
\limsup_{n\to\infty} \bbE_{P_{X^n\tilX_i^n}^{\bz_n}}\big[\Delta_i^{(n)}(X^n,\tilX_i^n)\big]&\le D_i,\\
\lim_{n\to\infty}\big\|P_{\tilX_i^n}^{\bz_n}-\psi_i^n\big\|_\mathrm{TV}&=0.
\end{align}


\subsubsection{Distortion and Perception Analyses}
\label{subapp:dis and per}

Firstly, we show that $\lim_{n\to\infty}\bbE_{\bZ_n}[T(\bZ_n)]=0$. It follows from the definition of
$Q^{\bz_n}_{X^n\tilX_1^n\tilX_2^nS_1S_2K}$ in
\eqref{eq:auxiliary-distribution} and the chain rule that
\begin{align}
&Q^{\bz_n}_{X^nKS_1S_2\tilX_1^n\tilX_2^n}
(x^n,\tilx_1^n,\tilx_2^n,s_1,s_2,k)\nn\\
&=Q_{X^nK}^{\bz_n}(x^n,k)Q_{S_1S_2|X^nK}^{\bz_n}(s_1,s_2|x^n,k)Q_{\tilX_1^n\tilX_2^n|X^nKS_1S_2}^{\bz_n}(\tilx_1^n,\tilx_2^n|x^n,k,s_1,s_2)\nn\\
&=Q_{X^nK}^{\bz_n}(x^n,k)f(s_1,s_2|x^n,k)
\phi_1(\tilx_1^n|s_1,k)\phi_2(\tilx_2^n|s_1,s_2,k),
\label{eq:Q to P}
\end{align}
where \eqref{eq:Q to P} follows from the fact that
\begin{align} 
Q_{S_1S_2|X^nK}^{\bz_n}(s_1,s_2|x^n,k)=\frac{Q_{X^nKS_1S_2}^{\bz_n}(x^n,k,s_1,s_2)}{Q_{X^nK}^{\bz_n}(x^n,k)}=\frac{P_{X|U_1U_2}^n(x^n|u_1^n(s_1,k),u_2^n(s_1,s_2,k))
}{\sum_{\tils_1\in[M_1],\tils_2\in[M_2]}
P_{X|U_1U_2}^n(x^n|u_1^n(\tils_1,k),u_2^n(\tils_1,\tils_2,k))
},
\end{align}
and the definition of $f(s_1,s_2|x^n,k)$ in \eqref{eq:encoder def} and the fact that given $(S_1,S_2,K)=(s_1,s_2,k)$, $(\tilX_1^n,\tilX_2^n)$ are generated conditionally independently according to the decoders $\phi_1$ and $\phi_2$, respectively, and are conditionally independent of $X^n$. It follows from \eqref{eq:actual-distribution} and \eqref{eq:Q to P} that $P_{S_1S_2\tilX_1^n\tilX_2^n|X^nK}^{\bz_n}=Q_{S_1S_2\tilX_1^n\tilX_2^n|X^nK}^{\bz_n}.$
Let $P_K\in\calP(\calK)$ be the distribution of the common randomness
$K$. Consequently, we obtain
\begin{align}
T({\bz_n})
&=\Big\|Q_{X^nKS_1S_2\tilX_1^n\tilX_2^n}^{\bz_n}-P_{X^nKS_1S_2\tilX_1^n\tilX_2^n}^{\bz_n}\Big\|_{\rm TV}\\
&=\Big\|Q_{X^nK}^{\bz_n}-P_X^nP_K\Big\|_{\rm TV}\label{eq:given K-0}\\
&=\frac{1}{M_{\rm c}}\sum_{k\in[M_{\rm c}]}\Big\|Q_{X^n|K=k}^{\bz_n}-P_X^n\Big\|_{\rm TV},
\label{eq:given K}
\end{align}
where \eqref{eq:given K-0} follows from~\cite[Lemma V.2]{Cuff_2013_Distributed_Channel_Synthesis}, \eqref{eq:given K} follows from the fact that $P_K(k)=Q_K^{\bz_n}(k)=\frac{1}{M_{\rm c}}$. Note that, for any fixed $k\in[M_{\rm c}]$ and each $x^n\in\calX^n$,
\begin{align}
Q_{X^n|K=k}^{\bz_n}(x^n)
=\frac{1}{M_1M_2}\sum_{s_1,s_2}
P_{X|U_1U_2}^n(x^n|u_1^n(s_1,k),u_2^n(s_1,s_2,k)).
\end{align}
This is precisely the output distribution obtained by selecting $(s_1,s_2)$ uniformly from the superposition codebook associated with $k$ and passing the corresponding codeword through $P_{X|U_1U_2}^n$. Recall from \eqref{general-region} that $R_1>I(X;U_1)$ and $R_2>I(X;U_1,U_2)$. It follows from the superposition soft-covering lemma \cite[Lemma 4]{Goldfeld_2020_Wiretap_Channels} and the Pinsker inequality \cite[Lemma 11.6.1]{Thomas_2006_Elements-of-information-theory} that, for every fixed $k\in[M_{\rm c}]$,
\begin{align}
\lim_{n\to\infty}\bbE_{\bZ_n}\Big[\big\|Q^{\bZ_n}_{X^n|K=k}-P_X^n\big\|_{\rm TV}\Big]=0.
\label{eq:soft approximate}
\end{align}
Combining \eqref{eq:given K} and \eqref{eq:soft approximate} yields
\begin{align}
\lim_{n\to\infty}\bbE_{\bZ_n}\big[T(\bZ_n)\big]=0.
\label{eq:q and p approx}
\end{align}

Secondly, we show that for any $i\in[2]$, $\bbE_{\bZ_n}[D_i^Q(\bZ_n)]\le D_i$.
For the first-layer reconstruction, averaging over the random codebook ensemble yields
\begin{align}
&\bbE_{\bZ_n}\Bigl[D_1^Q(\bZ_n)\Bigr]\nn\\*
&=\bbE_{\bZ_n}\Bigl[\sum_{x^n,\tilde{x}_1^n}\Delta_1^{(n)}(x^n,\tilde{x}_1^n)Q^{\bZ_n}_{X^n\tilde{X}_1^n}(x^n,\tilde{x}_1^n)\Bigr]\label{D1_1}\\*
&=\sum_{x^n,\tilde{x}_1^n}\Delta_1^{(n)}(x^n,\tilde{x}_1^n)\sum_{\tilde{x}_2^n,s_1,s_2,k}\bbE_{\bZ_n}\Bigl[Q^{\bZ_n}_{X^n\tilde X_1^n\tilde X_2^nS_1S_2K}(x^n,\tilde{x}_1^n,\tilde{x}_2^n,s_1,s_2,k)\Bigr]\label{D1_2}\\*
&=\frac{1}{M_1M_2M_{\rm c}}\sum_{x^n,\tilde{x}_1^n}\sum_{s_1,s_2,k}\Delta_1^{(n)}(x^n,\tilde{x}_1^n)\cdot\bbE_{\bZ_n}\Bigl[P_{X|U_1U_2}^n(x^n|U_1^n(s_1,k),U_2^n(s_1,s_2,k))P^n_{\tilde{X}_1|U_1}(\tilde{x}_1^n|U_1^n(s_1,k))\Bigr]\label{D1_3}\\*
&=\sum_{x^n,\tilde{x}_1^n,u_1^n,u_2^n}
\Delta_1^{(n)}(x^n,\tilde{x}_1^n)
P^n_{X\tilde{X}_1U_1U_2}
(x^n,\tilde{x}_1^n,u_1^n,u_2^n)
\label{D1_4}\\*
&=\sum_{x^n,\tilde{x}_1^n}
\Delta_1^{(n)}(x^n,\tilde{x}_1^n)P_{X\tilde{X}_1}^n(x^n,\tilde{x}_1^n)\label{D1_5}\\*
&=\bbE_{P_{X\tilX_1}}[\Delta_1(X,\tilde{X}_1)]\label{D1_8}\\*
&\leq D_1,\label{D1_9}
\end{align}
where \eqref{D1_1} follows from \eqref{eq:DQ-def}, \eqref{D1_2} follows by marginalization, \eqref{D1_3} follows from \eqref{eq:auxiliary-distribution}, \eqref{D1_4} follows from the fact that $(U_1^n(s_1,k),U_2^n(s_1,s_2,k))\sim P_{U_1U_2}^n$ for any $(s_1,s_2,k)\in[M_1]\times[M_2]\times[M_{\rm c}]$, \eqref{D1_5} follows by marginalizing over $(U_1^n,U_2^n)$, \eqref{D1_8} follows from \eqref{normalized-n-letter-distortion} and the fact that $(X_t,\tilX_{1,t})\sim P_{X\tilX_1}$ for any $t\in[n]$, and \eqref{D1_9} follows from the fact that $P_{XU_1U_2\tilX_1\tilX_2}\in\calM_{\rm SR}(D_1,D_2)$.
Similarly, we obtain
\begin{align}
\bbE_{\bZ_n}\Bigl[D_2^Q(\bZ_n)\Bigr]\leq D_2.\label{D_2}
\end{align}

Finally, we show that for any $i\in[2]$, $\lim_{n\to\infty}\bbE_{\bZ_n}[V_i^Q(\bZ_n)]=0$.
It follows from the definition of $Q^{\bz_n}_{X^n\tilX_1^n\tilX_2^nS_1S_2K}$ in \eqref{eq:auxiliary-distribution} that
\begin{align}
Q^{\bz_n}_{\tilde X_1^n}(\tilde{x}_1^n)
&=
\frac{1}{M_1M_{\rm c}}
\sum_{s_1,k}
P^n_{\tilde{X}_1|U_1}
(\tilde{x}_1^n|u_1^n(s_1,k)),
\label{eq:Q tilX1}\\
Q^{\bz_n}_{\tilde X_2^n}(\tilde{x}_2^n)
&=
\frac{1}{M_1M_2M_{\rm c}}
\sum_{s_1,s_2,k}
P^n_{\tilde{X}_2|U_1U_2}
(\tilde{x}_2^n|u_1^n(s_1,k),u_2^n(s_1,s_2,k)).
\label{eq:Q tilX2}
\end{align}
Let $\bZ_{1,n}:=\{U_1^n(s_1,k):(s_1,k)\in[M_1]\times[M_{\rm c}]\}$ be the first-layer codebook, consisting of $M_1M_{\rm c}$ mutually
independent codewords, each distributed according to $P_{U_1}^n$.
Consequently, \eqref{eq:Q tilX1} is precisely the output distribution
obtained by selecting a codeword uniformly from $\bZ_{1,n}$ and passing it
through $P_{\tilX_1|U_1}^n$.
Recall from \eqref{general-set} that $P_{\tilX_1}=\psi_1$, and from
\eqref{general-region} that
$R_1+R_{\rm c}>I(\tilde{X}_1;U_1)$.
It follows from the soft-covering lemma
\cite[Lemma IV.1]{Cuff_2013_Distributed_Channel_Synthesis} that
\begin{align}
\label{eq:tilX1 soft}
\lim_{n\to\infty}\bbE_{\bZ_n}\big[V_1^Q(\bZ_n)\big]=\lim_{n\to\infty}\bbE_{\bZ_n}\Big[\big\|Q^{\bZ_n}_{\tilde X_1^n} - \psi_1^n\big\|_{\rm{TV}}\Big]=0.
\end{align}
Similarly, let $\bZ_{2,n}:=\{(U_1^n(s_1,k),U_2^n(s_1,s_2,k)):(s_1,s_2,k)\in[M_1]\times[M_2]\times[M_{\rm c}]\}$ be the superposition codebook, consisting of $M_1M_\rmc$ mutually independent codewords, each distributed according to $P_{U_1}^n$, and $M_2$ conditionally independent codewords generated according to $P_{U_2|U_1}^n$. Consequently, \eqref{eq:Q tilX2} is precisely the output distribution obtained by selecting a codeword uniformly from $\bZ_{2,n}$ and passing it through $P_{\tilX_2|U_1U_2}^n$. Recall from \eqref{general-set} that $P_{\tilX_2}=\psi_2$, and from \eqref{general-region} that $R_1+R_{\rm c}>I(\tilX_2;U_1)$ and $R_2+R_{\rm c}>I(\tilX_2;U_1,U_2)$. It follows from the superposition soft-covering lemma \cite[Lemma 4]{Goldfeld_2020_Wiretap_Channels} and the Pinsker inequality \cite[Lemma 11.6.1]{Thomas_2006_Elements-of-information-theory} that
\begin{align}
\lim_{n\to\infty}\bbE_{\bZ_n}\big[V_2^Q(\bZ_n)\big]=\lim_{n\to\infty}\bbE_{\bZ_n}\Big[\big\|Q^{\bZ_n}_{\tilX_2^n}-\psi_2^n\big\|_{\rm TV}\Big]=0.\label{eq:tilX2 soft}
\end{align}
For any $i\in[2]$, combining \eqref{eq:tilX1 soft} and \eqref{eq:tilX2 soft} yields
\begin{align}
\lim_{n\to\infty}\bbE_{\bZ_n}\big[V_i^Q(\bZ_n)\big]=0.\label{eq:Q perception property}
\end{align}

\subsubsection{Existence of Deterministic Codebook}
For any $i\in[2]$, recall that $T(\bZ_n)$, $V_i^Q(\bZ_n)$ and
$D_i^Q(\bZ_n)$ are defined in \eqref{eq:T(C)-def}--\eqref{eq:DQ-def}. 
Fix any $\delta>0$ and define the bad event that a random codebook fails to simultaneously satisfy the desired distributional approximation and distortion bounds as
\begin{align}
\calB_{n,\delta}:=\{T(\bZ_n)>\delta\}\bigcup\big\{\cup_{i\in[2]}\{V_{i}^Q(\bZ_n)>\delta\}\big\}\bigcup\big\{\cup_{i\in[2]}\{D_{i}^Q(\bZ_n)>D_i+\delta\}\big\}.
\end{align}
It follows that, as $n\to\infty$,
\begin{align}
\Pr\{\calB_{n,\delta}\}&\le\frac{\bbE_{\bZ_n}[T(\bZ_n)]}{\delta}+\sum_{i\in[2]}\frac{\bbE_{\bZ_n}[V_i^Q(\bZ_n)]}{\delta}+\sum_{i\in[2]}\frac{\mathrm{Var}(D_i^Q(\bZ_n))}{\delta^2}\label{eq:bn pr to 0-1}\\
&=0,\label{eq:bn pr to 0-2}
\end{align}
where \eqref{eq:bn pr to 0-1} follows from the union bound, the Markov inequality and the Chebyshev inequality, which state that $\Pr\{Y>a\}\leq \frac{\bbE[Y]}{a}$ and $\Pr\{Y-\bbE[Y]>a\}\le\frac{\mathrm{Var}(Y)}{a^2}$ for any nonnegative random variable $Y$ and any $a>0$, respectively, and \eqref{eq:bn pr to 0-2} follows from \eqref{eq:q and p approx}, \eqref{eq:Q perception property} and the fact that $\lim_{n\to\infty}\mathrm{Var}(D_i^Q(\bZ_n))=0$, whose proof is deferred to the end of this part.
Consequently, choosing $\delta$ to decrease sufficiently slowly with $n$ yields a sequence of deterministic codebooks $\{\bz_n\}_{n\in\bbN}$ satisfying \eqref{eq:good code distortion} and \eqref{eq:good code perception} simultaneously. The proof of Lemma \ref{lem:good code} is completed.

\emph{Proof of $\lim_{n\to\infty}\mathrm{Var}(D_i^Q(\bZ_n))=0$}:
Define
\begin{align}
g_1(u_1,u_2)&:=\sum_{x,\tilx_1}P_{X|U_1U_2}(x|u_1,u_2)P_{\tilX_1|U_1}(\tilx_1|u_1)\Delta_1(x,\tilx_1),\\
g_2(u_1,u_2)&:=\sum_{x,\tilx_2}P_{X|U_1U_2}(x|u_1,u_2)P_{\tilX_2|U_1U_2}(\tilx_2|u_1,u_2)\Delta_2(x,\tilx_2).
\end{align}
Furthermore, for any $i\in[2]$, define
\begin{align}
G_i(u_1^n,u_2^n)&:=\frac{1}{n}\sum_{t\in[n]}g_i(u_{1,t},u_{2,t}).\label{eq:def of gi n}
\end{align}
Note that $G_i(u_1^n,u_2^n)$ for any $i\in[2]$ is the conditional average distortion given $(u_1^n,u_2^n)$.
Based on the above definitions, for any $i\in[2]$, it follows from \eqref{eq:auxiliary-distribution} and \eqref{eq:DQ-def} that
\begin{align}
D_i^Q(\bZ_n)
&=\frac{1}{M_1M_2M_\rmc}\sum_{s_1,s_2,k}G_i(U_1^n(s_1,k),U_2^n(s_1,s_2,k)).
\end{align}
Consequently, for any $i\in[2]$, we obtain that as $n\to\infty$,
\begin{align}
\mathrm{Var}(D_i^Q(\bZ_n))
&=\frac{1}{M_1^2M_2^2M_\rmc^2}\sum_{s_1,s_2,k}\sum_{\tils_1,\tils_2,\tilk}\mathrm{Cov}\Big(G_i\big(U_1^n(s_1,k),U_2^n(s_1,s_2,k)\big),G_i\big(U_1^n(\tils_1,\tilk),U_2^n(\tils_1,\tils_2,\tilk)\big)\Big)\label{eq:cal var dq-1}\\
&=\frac{1}{M_1M_2M_\rmc}\mathrm{Var}\Big(G_i\big(U_1^n(1,1),U_2^n(1,1,1)\big)\Big)\nn\\
&\qquad+\frac{M_2-1}{M_1M_2M_\rmc}\mathrm{Cov}\Big(G_i\big(U_1^n(1,1),U_2^n(1,1,1)\big),G_i(U_1^n(1,1),U_2^n(1,2,1))\Big)\label{eq:cal var dq-2}\\
&=\frac{1}{nM_1M_2M_\rmc}\mathrm{Var}\big(g_i(U_1,U_2)\big)+\frac{M_2-1}{nM_1M_2M_\rmc}\mathrm{Var}\big(\bbE[g_i(U_1,U_2)|U_1]\big)\label{eq:cal var dq-3}\\
&\le \frac{\Delta_{\max}^2}{4nM_1M_\rmc}\label{eq:cal var dq-4}\\
&=0,\label{eq:cal var dq-5}
\end{align}
where \eqref{eq:cal var dq-1} follows from the variance expansion, \eqref{eq:cal var dq-2} follows from the symmetry of the random codebook and the fact that the covariance vanishes whenever $(s_1,k)\neq(\tils_1,\tilk)$, \eqref{eq:cal var dq-3} follows from \eqref{eq:def of gi n} and the fact that $\{U_{1,t}(1,1)\}_{t\in[n]}$ are i.i.d. and $U_2^n(1,1,1)$, $U_2^n(1,2,1)$ are conditionally independent given $U_1^n(1,1)$, \eqref{eq:cal var dq-4} follows from the Popoviciu inequality on variances~\cite[Eq.~(4)]{Lim_2022_Geometrical-Bounds-for-Variance}, which states that $\mathrm{Var}(Y)\le \frac{(b-a)^2}{4}$ for any random variable $Y$ supported on $[a,b]$, and the fact that $g_i(u_1,u_2)\in[0,\Delta_{\max}]$ and $\bbE[g_i(U_1,U_2)|U_1]\in[0,\Delta_{\max}]$, and \eqref{eq:cal var dq-5} follows from the fact that $\Delta_{\max}$ is bounded, $(M_1,M_\rmc)\in\mathbb N^2$ and $n\to\infty$.

\bibliographystyle{IEEEtran}
\bibliography{Reference}
\end{document}